\documentclass{article}
\pdfoutput=1
\usepackage[numbers]{natbib}
\usepackage{etoolbox}

\usepackage[utf8]{inputenc}
\usepackage{comment}

\usepackage{multirow}
\usepackage{hhline}

\usepackage{microtype}
\usepackage{amsmath,amssymb,amsthm}
\usepackage{mathtools}
\usepackage{graphicx}
\usepackage{tikz}
\usepackage{authblk}

\usetikzlibrary{decorations,decorations.pathreplacing,arrows,petri,topaths,backgrounds,shapes,positioning,fit,calc,matrix,patterns}

\newcommand{\pt}{po\-ly\-nom\-i\-al-time}

\usepackage[colorlinks,citecolor=green!60!black,linkcolor=red!60!black,pagebackref]{hyperref}
\usepackage{cleveref}

\newcounter{saveenum}

\newtheorem{proposition}{Proposition}
\newtheorem{theorem}{Theorem}
\newtheorem{lemma}{Lemma}
\newtheorem{claim}{Claim}
\newtheorem{corollary}{Corollary}

\theoremstyle{definition}

\newtheorem{definition}{Definition}

\crefname{table}{Table}{Tables}
\crefname{figure}{Figure}{Figures}
\crefname{theorem}{Theorem}{Theorems}
\crefname{definition}{Definition}{Definitions}
\crefname{corollary}{Corollary}{Corollaries}
\crefname{obs}{Observation}{Observations}
\crefname{lemma}{Lemma}{Lemmas}
\crefname{example}{Example}{Examples}
\crefname{claim}{Claim}{Claims}
\crefname{proposition}{Proposition}{Propositions}
\crefname{construction}{Construction}{Constructions}
\crefname{property}{Property}{Properties}

\crefname{subsection}{Subsection}{Subsections}
\crefname{section}{Section}{Sections}
\crefname{algorithm}{Algorithm}{Algorithms}
\crefname{section}{Section}{Sections}
\crefname{subsection}{Section}{Sections}
\crefname{figure}{Figure}{Figures}
\crefname{theorem}{Theorem}{Theorems}
\crefname{lemma}{Lemma}{Lemmas}
\crefname{proposition}{Proposition}{Propositions}
\crefname{definition}{Definition}{Definitions}
\crefname{claim}{Claim}{Claims}

\newcommand{\NP}{NP\xspace}
\newcommand{\classP}{P\xspace}

\usepackage{paralist}
\usepackage{xspace}

\DeclareMathOperator{\card}{card}

\DeclareMathOperator{\new}{new}
\newcommand{\sold}{\ensuremath{s}}
\newcommand{\snew}{\ensuremath{s_{\new}}}
\newcommand{\Deltas}{\ensuremath{\Delta_{s}}}
\newcommand{\kwin}{\ensuremath{r_{\alpha}}}
\newcommand{\kdis}{\ensuremath{d}}
\newcommand{\krem}{\ensuremath{r}}

\newcommand{\demand}{\ensuremath{\kappa}}

\newcommand{\partyA}{A\xspace}
\newcommand{\partyB}{B\xspace}
\newcommand{\partyAsupporter}{\partyA-supporter\xspace}
\newcommand{\partyAsupporters}{\partyA-supporters\xspace}
\newcommand{\partyBsupporter}{\partyB-supporter\xspace}
\newcommand{\partyBsupporters}{\partyB-supporters\xspace}

\DeclareMathOperator\inc{inc}
\DeclareMathOperator\adj{adj}

\newcommand{\Vdis}{\ensuremath{D}}
\newcommand{\Vrem}{\ensuremath{R}}
\newcommand{\Vwin}{\ensuremath{R_\alpha}}

 \newcommand{\decprob}[3]{%
 {\def\descriptionlabel##1{\hspace\labelsep\quad{}\it{}##1}%
 \par\vspace{\topsep}\noindent%
 \begin{compactdesc}
 \item[\textsc{#1}]
 \item[Input:] #2
 \item[Question:] #3
\end{compactdesc}
}\vspace{\topsep}}

\newcommand{\probValid}{\textsc{Dissolution}\xspace}
\newcommand{\probWinner}{\textsc{Biased Dissolution}\xspace}
\newcommand{\probExCover}{\textsc{Exact Cover by $t$-Sets}\xspace}

\newcommand{\probPlanarPacking}[1]{\textsc{Perfect Planar #1-Matching}\xspace}
\newcommand{\probRestTwoFactor}{\textsc{$L$-Restricted Two Factor}\xspace}

\newcommand{\starpartition}[1]{\ensuremath{#1}-star partition\xspace}

\newcommand{\dissolution}{dissolution\xspace}
\newcommand{\biased}{biased\xspace}
\newcommand{\biaseddissolution}{\biased \dissolution}

\newcommand{\Sdissol}[2]{\mbox{\ensuremath{(#1,#2)}}-\dissolution}
\newcommand{\AbiasedSdissol}[3]{$#1$-\biased \Sdissol{#2}{#3}}

\newcommand{\biasedSdissol}[2]{\kwin-\biased \Sdissol{#1}{#2}}

\newcommand{\dissol}{\Sdissol{\sold}{\Deltas}}
\newcommand{\biaseddissol}{\AbiasedSdissol{\kwin}{\sold}{\Deltas}}

\newcommand{\A}[1]{\ensuremath{Z(#1)}} %

\newcommand{\move}{move\xspace}
\newcommand{\moved}{moved\xspace}

\newcommand{\moving}{moving\xspace}

\newcommand{\z}{\ensuremath{z}}
\newcommand{\w}{\ensuremath{\alpha}}
\newcommand{\za}{\ensuremath{\z_{\w}}}

\tikzset{Asupporter/.style={draw,shape=circle,fill=black,scale=.5}}
\tikzset{Bsupporter/.style={draw,shape=circle,fill=white,scale=.5}}

\tikzset{ZEdge/.style={thick,-,sloped}}
\tikzset{ZArrow/.style={thick,->,sloped}}

\newcommand{\nodeBBBBB}[2]{
    \node [Bsupporter,#2] (#1v1) [transform shape] {};
    \node [Bsupporter, below left = 4.5pt of #1v1] (#1v2) [transform shape] {};
    \node [Bsupporter, below right = 4.5pt of #1v1] (#1v3) [transform shape] {};
    \node [Bsupporter, below right = 4.5pt of #1v2] (#1v4) [transform shape] {};
    \node [fit=(#1v1)(#1v2)(#1v3)(#1v4),draw, inner sep=0pt, circle] (#1) [transform shape] {};
    \node [Bsupporter, below right = 0pt of #1.center,anchor=center] (#1v5) [transform shape] {};
}

\newcommand{\nodeABBBB}[2]{
    \node [Asupporter,#2] (#1v1) [transform shape] {};
    \node [Bsupporter, below left = 4.5pt of #1v1] (#1v2) [transform shape] {};
    \node [Bsupporter, below right = 4.5pt of #1v1] (#1v3) [transform shape] {};
    \node [Bsupporter, below right = 4.5pt of #1v2] (#1v4) [transform shape] {};
    \node [fit=(#1v1)(#1v2)(#1v3)(#1v4),draw, inner sep=0pt, circle] (#1) [transform shape] {};
    \node [Bsupporter, below right = 0pt of #1.center,anchor=center] (#1v5) [transform shape] {};
}

\newcommand{\nodeAABBB}[2]{
    \node [Asupporter,#2] (#1v1) [transform shape] {};
    \node [Bsupporter, below left = 4.5pt of #1v1] (#1v2) [transform shape] {};
    \node [Bsupporter, below right = 4.5pt of #1v1] (#1v3) [transform shape] {};
    \node [Asupporter, below right = 4.5pt of #1v2] (#1v4) [transform shape] {};
    \node [fit=(#1v1)(#1v2)(#1v3)(#1v4),draw, inner sep=0pt, circle] (#1) [transform shape] {};
    \node [Bsupporter, below right = 0pt of #1.center,anchor=center] (#1v5) [transform shape] {};
}

\newcommand{\nodeAAABB}[2]{
    \node [Asupporter,#2] (#1v1) [transform shape] {};
    \node [Bsupporter, below left = 4.5pt of #1v1] (#1v2) [transform shape] {};
    \node [Bsupporter, below right = 4.5pt of #1v1] (#1v3) [transform shape] {};
    \node [Asupporter, below right = 4.5pt of #1v2] (#1v4) [transform shape] {};
    \node [fit=(#1v1)(#1v2)(#1v3)(#1v4),draw, inner sep=0pt, circle] (#1) [transform shape] {};
    \node [Asupporter, below right = 0pt of #1.center,anchor=center] (#1v5) [transform shape] {};
}

\newcommand{\nodeAAA}[2]{
    \node [Asupporter,#2] (#1v1) [transform shape] {};
    \node [Asupporter, below left = 4.5pt of #1v1] (#1v2) [transform shape] {};
    \node [Asupporter, below right = 4.5pt of #1v1] (#1v3) [transform shape] {};
    \node [fit=(#1v1)(#1v2)(#1v3),draw, inner sep=0pt, circle] (#1) [transform shape] {};
}

\newcommand{\nodeAAAd}[2]{
    \node [Asupporter,#2] (#1v1) [transform shape] {};
    \node [Asupporter, below left = 4.5pt of #1v1] (#1v2) [transform shape] {};
    \node [Asupporter, below right = 4.5pt of #1v1] (#1v3) [transform shape] {};
    \node [fit=(#1v1)(#1v2)(#1v3),draw,dashed, inner sep=0pt, circle] (#1) [transform shape] {};
}

\newcommand{\nodeAAB}[2]{
    \node [Asupporter,#2] (#1v1) [transform shape] {};
    \node [Asupporter, below left = 4.5pt of #1v1] (#1v2) [transform shape] {};
    \node [Bsupporter, below right = 4.5pt of #1v1] (#1v3) [transform shape] {};
    \node [fit=(#1v1)(#1v2)(#1v3),draw, inner sep=0pt, circle] (#1) [transform shape] {};
}

\newcommand{\nodeABB}[2]{
    \node [Asupporter,#2] (#1v1) [transform shape] {};
    \node [Bsupporter, below left = 4.5pt of #1v1] (#1v2) [transform shape] {};
    \node [Bsupporter, below right = 4.5pt of #1v1] (#1v3) [transform shape] {};
    \node [fit=(#1v1)(#1v2)(#1v3),draw, inner sep=0pt, circle] (#1) [transform shape] {};
}

\newcommand{\nodeABBd}[2]{
    \node [Asupporter,#2] (#1v1) [transform shape] {};
    \node [Bsupporter, below left = 4.5pt of #1v1] (#1v2) [transform shape] {};
    \node [Bsupporter, below right = 4.5pt of #1v1] (#1v3) [transform shape] {};
    \node [fit=(#1v1)(#1v2)(#1v3),draw,dashed,inner sep=0pt, circle] (#1) [transform shape] {};
}

\newcommand{\nodeBBB}[2]{
    \node [Bsupporter,#2] (#1v1) [transform shape] {};
    \node [Bsupporter, below left = 4.5pt of #1v1] (#1v2) [transform shape] {};
    \node [Bsupporter, below right = 4.5pt of #1v1] (#1v3) [transform shape] {};
    \node [fit=(#1v1)(#1v2)(#1v3),draw, inner sep=0pt, circle] (#1) [transform shape] {};
}

\newcommand{\nodeBBBd}[2]{
    \node [Bsupporter,#2] (#1v1) [transform shape] {};
    \node [Bsupporter, below left = 4.5pt of #1v1] (#1v2) [transform shape] {};
    \node [Bsupporter, below right = 4.5pt of #1v1] (#1v3) [transform shape] {};
    \node [fit=(#1v1)(#1v2)(#1v3),draw,dashed,inner sep=0pt, circle] (#1) [transform shape] {};
}

\newcommand{\nodeAxx}[2]{
    \node [Asupporter,#2] (#1v1) [transform shape] {};
    \node [scale=.5,below left = 4pt of #1v1] (#1v2) [transform shape] {x};
    \node [scale=.5,below right = 4pt of #1v1] (#1v3) [transform shape] {x};
    \node [fit=(#1v1)(#1v2)(#1v3),draw, inner sep=0pt, circle] (#1) [transform shape] {};
}

\newcommand{\nodeAB}[2]{
    \node [Asupporter,#2] (#1v1) [transform shape] {};
    \node [Bsupporter, right = 4pt of #1v1] (#1v2) [transform shape] {};
    \node [fit=(#1v1)(#1v2),draw, inner sep=0.5pt, circle] (#1) [transform shape] {};
}

\newcommand{\nodeABd}[2]{
    \node [Asupporter,#2] (#1v1) [transform shape] {};
    \node [Bsupporter, right = 4pt of #1v1] (#1v2) [transform shape] {};
    \node [fit=(#1v1)(#1v2),draw,dashed,inner sep=0.5pt, circle] (#1) [transform shape] {};
}

\newcommand{\nodeBB}[2]{
    \node [Bsupporter,#2] (#1v1) [transform shape] {};
    \node [Bsupporter, right = 4pt of #1v1] (#1v2) [transform shape] {};
    \node [fit=(#1v1)(#1v2),draw, inner sep=0.5pt, circle] (#1) [transform shape] {};
}

\newcommand{\nodeBBd}[2]{
    \node [Bsupporter,#2] (#1v1) [transform shape] {};
    \node [Bsupporter, right = 4pt of #1v1] (#1v2) [transform shape] {};
    \node [fit=(#1v1)(#1v2),draw,dashed,inner sep=0.5pt, circle] (#1) [transform shape] {};
}

\newcommand{\lA}[1]{
  \begin{tikzpicture}[scale=#1]
    \node [Asupporter] (a) [transform shape] {};
  \end{tikzpicture}
}
\newcommand{\lB}[1]{
  \begin{tikzpicture}[scale=#1]
    \node [Bsupporter] (b) [transform shape] {};
  \end{tikzpicture}
}

\newcommand{\lAA}[1]{
  \begin{tikzpicture}[scale=#1]
    \node [Asupporter] (a1) [transform shape] {};
    \node [Asupporter, right = 1pt of a1] (a2) [transform shape] {};
  \end{tikzpicture}
}
\newcommand{\lBB}[1]{
  \begin{tikzpicture}[scale=#1]
    \node [Bsupporter] (b1) [transform shape] {};
    \node [Bsupporter, right = 1pt of b1] (b2) [transform shape] {};
  \end{tikzpicture}
}
\newcommand{\lAB}[1]{
  \begin{tikzpicture}[scale=#1]
    \node [Asupporter] (a) [transform shape] {};
    \node [Bsupporter, right = 1pt of a] (b) [transform shape] {};
  \end{tikzpicture}
}

\tikzset{Edge/.style={thick,-,shorten <= -8pt, shorten >= -8pt,sloped}}
\tikzset{SEdge/.style={thick,-,shorten <= -2.5pt, shorten >= -2.5pt,sloped}}
\tikzset{Arrow/.style={thick,->,shorten <= -8pt, shorten >= -8pt,sloped}}
\tikzset{SArrow/.style={thick,->,shorten <= -2.5pt, shorten >= -2.5pt,sloped}}

\newcommand{\moveA}{
  \lA{1.0}
}
\newcommand{\moveB}{
  \lB{1.0}
}
\newcommand{\moveAA}{
  \lAA{1.0}
}
\newcommand{\moveBB}{
  \lBB{1.0}
}
\newcommand{\moveAB}{
  \lAB{1.0}
}

\newcommand{\networkI}{\ensuremath{I^*}}
\newcommand{\networkG}{\ensuremath{G^*}}
\newcommand{\networkV}{\ensuremath{V^*}}
\newcommand{\networkE}{\ensuremath{E^*}}
\newcommand{\networkCap}{\ensuremath{c^*}}
\newcommand{\networks}{\ensuremath{\sigma}}
\newcommand{\networkt}{\ensuremath{\tau}}
\newcommand{\networkInst}{\ensuremath{(\networkG=(\networkV,\networkE), \networkCap,\networks,\networkt)}}
\title{Network-Based Vertex Dissolution\thanks{An extended abstract appeared under the title ``Network-Based Dissolution'' in the Proceedings of the 39th International Symposium on Mathematical Foundations of Computer Science (MFCS~'14), volume 8635 of LNCS, Springer, pages 69--80.}}

\author[1]{René van Bevern}
\author[1]{Robert Bredereck}
\author[1]{Jiehua Chen}
\author[1]{Vincent~Froese}
\author[1]{Rolf~Niedermeier}
\author[2]{Gerhard~J.~Woeginger}

\affil[1]{Institut f\"ur Softwaretechnik und Theoretische Informatik, TU Berlin, Germany. \texttt{\{rene.vanbevern, robert.bredereck, jiehua.chen, vincent.froese, rolf.niedermeier\}@tu-berlin.de}}

\affil[2]{Department of Mathematics and Computer Science,
TU Eindhoven, The~Netherlands. \texttt{gwoegi@win.tue.nl}}

\date{}

\begin{document}
\def\sectionautorefname{Section}
\def\subsectionautorefname{Section}
\maketitle

\begin{abstract}
We introduce a graph-theoretic vertex dissolution model that applies to a number 
of redistribution scenarios such as gerrymandering in political districting 
or work balancing in an online situation.
The central aspect of our model is the deletion of certain vertices and the 
redistribution of their load to neighboring vertices in a completely balanced way. 

We investigate how the underlying graph structure, the knowledge of which 
vertices should be deleted, and the relation between old and new vertex loads 
influence the computational complexity of the underlying graph problems. 
Our results establish a clear borderline between tractable and intractable cases.
\end{abstract}

\section{Introduction}
Motivated by applications in areas like political redistricting, economization, 
and distributed systems, we introduce a class of graph modification problems 
that we call \emph{network-based vertex dissolution}. 
We are given an undirected graph where each vertex carries a load 
consisting of discrete entities (e.\,g.\ voters, tasks, data).
These loads are \emph{balanced}: all vertices carry the same load.
Now a certain number of vertices has to be \emph{dissolved}, that is, they 
are to be deleted from the graph and their loads are to be redistributed 
among their neighbors so that afterwards all loads are balanced again. 

In fact, our vertex dissolution problem comes in two flavors: \probValid{}
and \probWinner{}.
\probValid{} is the basic version described in the preceding paragraph.
\probWinner{} is a variant that is motivated by gerrymandering in the context 
of political districting.
It is centered around a two-party scenario with two types, A and~B, of
discrete entities.
The goal is to find a redistribution that maximizes the number of vertices
in which the A-entities form a majority.
See Section~\ref{sect:prelim} for a formal definition of these models.

Our focus lies on analyzing the computational complexity of 
network-based vertex dissolution problems and on getting a good understanding 
of \pt{} solvable and \NP-hard cases.

\subsection{Three application scenarios}
We discuss three example scenarios in detail.
The first two examples more relate to \probWinner{}, while the third
example is closer to \probValid{}.

Our first example comes from \emph{political districting}, the process
of setting electoral districts.  Let us consider a situation with two
political parties (\partyA and~\partyB) and an electorate of voters
that each support either~\partyA or~\partyB.  The electorate is
currently divided into $n$~districts, each consisting of precisely
$\sold$~individual voters.  A district is won by the party that
receives the majority of votes in the district (for simplification,
assume that ties are resolved in favor of \partyB).  The local
government performs an electoral reform that reduces the number of
districts, and the local governor (from party~\partyA) is in charge of
the redistricting.  His goal is, of course, to let party~{\partyA} win
as many districts as possible while dissolving some districts and
moving their voters to adjacent districts.  All resulting new
districts should have equal size~$\snew$ (where $\snew>\sold$).  In
the network-based vertex dissolution model, the districts and their
neighborhoods are represented by an undirected graph, where each
vertex represents a district and each edge indicates that two
districts are adjacent.

Our second example concerns storage updates in parallel or distributed
systems.  Consider a distributed storage array consisting of
$n$~storage nodes, each having a capacity of~$\sold{}$ storage units,
of which some units are empty.  As the prices on cheap hard disk space
are rapidly decreasing, the operators want to upgrade the storage
capacity of some nodes and to deactivate other nodes for saving energy
and cost.  As their distributed storage concept takes full advantage
only if all nodes have equal capacity, they want to upgrade all
(non-deactivated) nodes to the same capacity~$\snew{}$ and move
capacities from deactivated nodes to non-deactivated neighboring
nodes.  In the resulting system, every non-deactivated node should
only use half of its storage units.  In the network-based vertex
dissolution model, storage nodes and their neighborhoods are
represented by an undirected graph, where each vertex represents a
storage node and each edge indicates that two nodes are neighbored in the
array.  The storage units are modeled by our two-party variant, where
empty units are represented by party~A and used units are represented
by party~B.  Finally, one asks for redistribution such that A-entities
form a majority for every vertex.

\looseness=-1 Our third and last example concerns \emph{economization} in a fairly general form. 
Consider a company with $n$~employees, each producing $\sold{}$~units 
of a desirable good during a month; for concreteness, let 
us say that each employee proves $\sold{}$~theorems per month.  
Now, due to the increasing support of automatic theorem provers, each employee 
is able to prove $\snew{}$~theorems per month ($\snew>\sold$). 
Hence, without lowering the total number of proved theorems per month, some 
employees may be moved to a special task force for improving automatic 
theorem provers: this will secure the company's future competitiveness in 
proving theorems, without decreasing the overall theorem output. 
By company regulations, all theorem-proving employees have to be treated
equally and should have identical workload.
In the network-based vertex dissolution model employees correspond to
vertices and edges indicate that the corresponding employees are comparable
in qualification and research interest.
Employees in the special task force are dissolved and disappear from the
scene of action; their workload is to be taken over by 
employees who are comparable in qualification and research interests.

\subsection{Related work}
We are not aware of any previous work on our network-based vertex
dissolution problem.  Our main inspiration came from the area of
political districting, in particular from
gerrymandering~\cite{LS14,PT08,PT09}, and from supervised
regionalization methods~\cite{DRS07}.  Of course, graph-theoretic
models have been employed earlier for (political) districting; for
instance, \citet{MJN98} draws a connection to graph partitioning, and
\citet{Duq04} and \citet{MS95} use graphs to model geographic
information in the regionalization problem.  These models are tailored
towards very specific applications and are mainly used for the purpose
of developing efficient heuristic algorithms, often relying on
mathematical programming techniques.  The computational hardness of
districting problems has been known for quite some time~\cite{Alt98}.

\looseness=-1 Also related to our problem is constructive (or destructive) control
by partitioning voters which has been introduced by \citet{BTT92}.  In
this scenario, a chair wants to make some candidate become a winner
(or a looser) by partitioning the set of voters and applying some
multi-stage voting protocol.  The crucial difference to our model is
that there are no restrictions on the possible voter set partitions.
The computational complexity of control by voter partitioning has been
investigated for many voting rules (\citet{FaliszewskiRotheChapter}
give an overview).

\subsection{Remark on nomenclature}
For the ease of presentation, throughout the work we will adopt a political
districting point of view on network-based vertex dissolution: the words districts
and vertices are used interchangeably, and the entities in districts are
referred to as voters or supporters.

\subsection{Contributions and organization of this paper}
\looseness=-1 We propose two simple computational problems \probValid{} and \probWinner{} to make the model for network-based vertex dissolution (\cref{sect:prelim}) concrete.  In the main body of our work, we provide a variety of computational tractability and intractability results for both problems.  We investigate relations of our new modeling to established models like matchings and flow networks.  Furthermore, we analyze how the structure of the underlying graphs or how an in-advance fixing of which vertices should be dissolved influences the computational complexity (mainly in terms of \pt{} solvability versus \NP-hard cases).

In \cref{sec:relations}, using flow networks, we show that \probWinner is \pt{}
solvable if the set of districts to be dissolved and the set of districts 
to be won are both specified as part of the input.
Furthermore, we show how our new model generalizes established models
such as partitioning graphs into stars and perfect matchings.

\cref{sec:Dichotomy} presents a complexity dichotomy for \probValid{} 
and \probWinner{} with respect to the old district size $\sold$ and the increase 
$\Deltas$ in district size (that is, the difference between the new and the old district size).
\probValid{} is \pt{} solvable for~$\sold=\Deltas$ and \probWinner{}
is \pt{} solvable for~$\sold=\Deltas=1$; all other cases are \NP-complete.

\cref{sec:graph-classes} analyzes the complexity of \probValid{} and 
\probWinner{} for various specially structured graphs, including
planar graphs (\NP-complete),
cliques (\pt{} solvable),
and graphs of bounded treewidth (linear-time solvable if $\sold$ and $\Deltas$ are constant).

\section{Formal setting}
\label{sect:prelim}

\newcommand{\constToNeighbors}
{\ensuremath{\forall {v' \in \Vdis, v \in V\setminus N(v')}: \z(v',v)=0}}
\newcommand{\constFull}
{\ensuremath{\forall {v' \in \Vdis}: \smashoperator{\sum_{(v',v)\in \A{\Vdis,G}}} \z(v',v)=\sold}}
\newcommand{\constEqualSize}
{\ensuremath{\forall {v \in V \setminus \Vdis}: \smashoperator{\sum_{(v',v)\in \A{\Vdis,G}}} \z(v',v)=\Deltas}}

\newcommand{\constAsupportersUpperbound}
{\ensuremath{\forall (v',v) \in \A{\Vdis,G}: \za(v',v) \le \z(v',v)}}
\newcommand{\constAsupportersFull}
{\ensuremath{\forall v' \in \Vdis: \smashoperator{\sum_{(v',v)\in \A{\Vdis,G}}} \za(v',v)= \w(v')}}
\newcommand{\constWinningdistricts}
{\ensuremath{\forall {v \in \Vwin}: 
    \w(v) + \smashoperator{\sum_{(v',v) \in \A{\Vdis,G}}}{\za(v',v)} > \frac{\sold+\Deltas}{2}}}

We start by introducing notation and formal definitions of the
technical terms that we use throughout the paper.

\subsection{Graphs}
Unless stated otherwise, we consider simple, undirected graphs~$G=(V,E)$, where~$V$ is a set
of $n$~vertices and~$E \subseteq {V \choose 2}$ is a set of~$m$ edges.
We use ${V \choose 2}$ to denote the family of all size-two subsets of~$V$.
For a given graph~$G$, we denote by~$V(G)$ the set
of vertices and by~$E(G)$ the set of edges of~$G$.
For a subset~$V' \subseteq V(G)$ of vertices and a
subset~$E'\subseteq (E(G)\cap {V'\choose 2})$ of edges,
the graph $G'=(V',E')$ is called a \emph{subgraph} of~$G$.
We also say $G$ contains~$G'$.
For a vertex subset~$V'\subseteq V$, the \emph{induced
subgraph}~$G[V']$ of~$G$ is defined as~$G[V']:=(V',E\cap {V'\choose 2})$.

A \emph{path} is a graph~$P=(V,E)$ with vertex set~$V=\{v_1,v_2,\dots,v_n\}$
and edge set $E=\{\{v_1,v_2\},\{v_2,v_3\},\dots,\{v_{n-1},v_n\}\}$.
The vertices~$v_1$ and~$v_n$ are the \emph{endpoints} of the path~$P$.
We say two vertices~$v$ and $v'$ in a graph~$G$ are \emph{connected} if $G$~contains
a path with the endpoints~$v$ and~$v'$.
A graph is \emph{connected} if every two vertices are connected.
The \emph{connected components} of a graph are its maximal connected subgraphs.
For a vertex~$v\in V$, we denote by~$N(v):=\{u\in V \mid \{u,v\}\in E\}$ 
the \emph{(open) neighborhood} of~$v$, that is, all vertices that are
connected to~$v$ by an edge.

A \emph{$t$-star} is a graph~$K_{1,t}=(V,E)$ with vertex set~$V=\{v_1,v_2,\dots,v_{t+1}\}$
and edge set $E=\{\{v_1,v_i\} \mid 2 \le i \le t+1\}$.
The vertex~$v_1$ is called the \emph{center} of the star.
A \emph{$t$-star partition} of~$G$ is a
partition~$\{V_1,\ldots,V_{n/(t+1)}\}$ of the vertex set~$V$ into
subsets of size~$t+1$ such that each~$G[V_i]$ contains a~$t$-star as a
subgraph.
Note that a $1$-star partition is a \emph{perfect matching}.

\subsection{Networks and flows}
A flow network~$\networkI$ consists of a directed graph~$\networkG=(\networkV,\networkE)$,
where $\networkV$ is the set of nodes and $\networkE$ is a set of arcs, 
an arc capacity function~$\networkCap\colon\networkE\to \mathbb{R}^+$, and two distinguished nodes~$\networks, \networkt\in \networkV$ called the \emph{source} and the \emph{target} of the network.
An arc is an ordered pair of nodes from~$\networkV$ and $\mathbb{R}^+$
is the set of non-negative real numbers.

A $(\networks,\networkt)$-flow~$f\colon\networkE \to \mathbb{R}^+$ is an arc value function with $f(u,v) \ge 0$ for all $(u,v)\in \networkE$ such that
\begin{enumerate}
 \item the \emph{capacity constraint} is fulfilled, i.e.,  
       $$\forall (u,v)\in \networkE: f(u,v) \le c(u,v)\text{, and}$$
 \item the \emph{conservation property} is satisfied, i.e., 
       $$\forall u \in \networkV\setminus \{\networks,\networkt\}: \smashoperator{\sum_{(u,v)\in \networkE}}f(u,v) = \smashoperator{\sum_{(v,u)\in \networkE}}f(v,u).$$
\end{enumerate}
We call $f$ \emph{integer} if all its values are integers. 
The \emph{value} of~$f$ is $\sum_{(\networks,u)\in \networkE}f(\networks,u)$.
Note that we distinguish between vertices in graphs and nodes in flow networks.

\subsection{Dissolutions}
Let $G$ be an undirected graph representing $n$~districts.
Let $\sold,\Deltas \in \mathbb{N}^+$~be the \emph{district size} and
\emph{district size increase}, respectively,
where $\mathbb{N}^+$ is the set of non-negative integer numbers.
For a subset~$V'\subseteq V(G)$ of districts, let
$$\A{V',G}:=\{(x,y) \mid x\in V' \wedge y \in V(G) \setminus V' \wedge \{x,y\} \in E(G)\}$$
be the set of district pairs consisting of a district from~$V'$ and
a neighbor that is not from~$V'$.
The central notion for our studies is that of 
a \emph{\dissolution}, which basically describes a valid movement of voters
from dissolved districts into non-dissolved districts.
The formal definition is as follows:

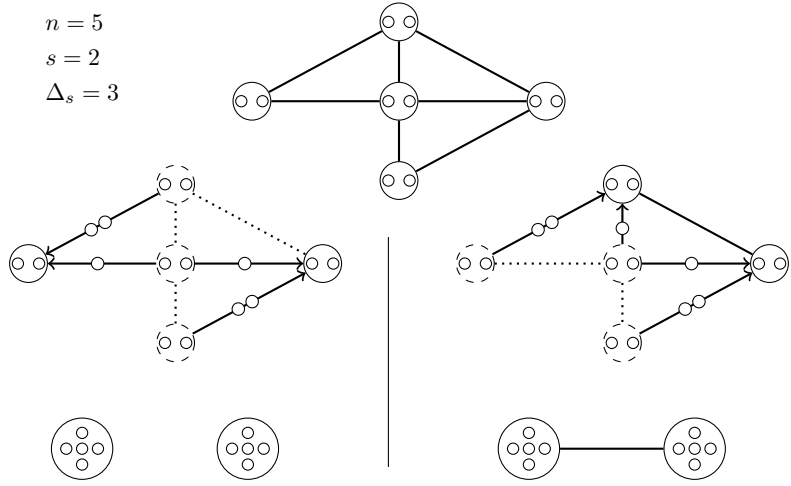
\begin{figure}[t!]
  \centering
    \begin{tikzpicture}[scale=0.9]
      \begin{scope}
        \node at (2,3.6) (null) {};
        \nodeBB{v1}{below=0pt of null}
        \nodeBB{v3}{below=1 of v1v1}
        \nodeBB{v2}{left=2.0 of v3v1}
        \nodeBB{v4}{right=2.0 of v3v1}
        \nodeBB{v5}{below=1 of v3v1}

        \node [left = 4cm of v1v1] (n) [transform shape]  {$n=5$};
        \node[below = 15pt of n.west,anchor=west] (sold) [transform shape]  {$\sold=2$};
        \node[below = 15pt of sold.west,anchor=west] (snew) [transform shape]  {$\Deltas=3$};
        \draw [ZEdge] (v1) -- (v2);
        \draw [ZEdge] (v1) -- (v3);
        \draw [ZEdge] (v1) -- (v4);
        \draw [ZEdge] (v2) -- (v3);
        \draw [ZEdge] (v3) -- (v4);
        \draw [ZEdge] (v3) -- (v5);
        \draw [ZEdge] (v4) -- (v5);
        
        \node [below=0.5cm of v5v1] (sepStart) [transform shape] {};
        \node [below=3.4cm of sepStart] (sepEnd) [transform shape] {};
         
        \draw [transform shape] (sepStart) -- (sepEnd);
      \end{scope}
      
      \begin{scope}[xshift=-3.3cm, yshift=-2.4cm]
        \node at (2,3.6) (null) {};
        \nodeBBd{v1}{below=0pt of null}
        \nodeBBd{v3}{below=1 of v1v1}
        \nodeBB{v2}{left=2.0 of v3v1}
        \nodeBB{v4}{right=2.0 of v3v1}
        \nodeBBd{v5}{below=1 of v3v1}
        \draw [ZEdge,dotted] (v1) -- (v4);
        \draw [ZEdge,dotted] (v1) -- (v3);
        \draw [ZEdge,dotted] (v5) -- (v3);
        \draw [ZArrow] (v1) -- (v2) node [midway]{\moveBB};
        \draw [ZArrow] (v3) -- (v2) node [midway]{\moveB};
        \draw [ZArrow] (v3) -- (v4) node [midway]{\moveB};
        \draw [ZArrow] (v5) -- (v4) node [midway]{\moveBB};
      \end{scope}
      
      \begin{scope}[xshift=3.3cm, yshift=-2.4cm]
        \node at (2,3.6) (null) {};
        \nodeBB{v1}{below=0pt of null}
        \nodeBBd{v3}{below=1 of v1v1}
        \nodeBBd{v2}{left=2.0 of v3v1}
        \nodeBB{v4}{right=2.0 of v3v1}
        \nodeBBd{v5}{below=1 of v3v1}
        \draw [ZEdge] (v1) -- (v4);
        \draw [ZEdge,dotted] (v2) -- (v3);
        \draw [ZEdge,dotted] (v3) -- (v5);
        \draw [ZArrow] (v2) -- (v1) node [midway]{\moveBB};
        \draw [ZArrow] (v3) -- (v1) node [midway]{\moveB};
        \draw [ZArrow] (v3) -- (v4) node [midway]{\moveB};
        \draw [ZArrow] (v5) -- (v4) node [midway]{\moveBB};
      \end{scope}

      \begin{scope}[xshift=-3.3cm, yshift=-6.3cm]
        \node at (2,3.6) (null) {};
        \nodeBBBBB{v1}{left=1 of null}
        \nodeBBBBB{v2}{right=1 of null}
      \end{scope}
      
      \begin{scope}[xshift=3.3cm, yshift=-6.3cm]
        \node at (2,3.6) (null) {};
        \nodeBBBBB{v1}{left=1 of null}
        \nodeBBBBB{v2}{right=1 of null}
        \draw [ZEdge] (v1) -- (v2);
      \end{scope}
  \end{tikzpicture}
  \caption{An illustration of two \Sdissol{2}{3}{}s.
    Small circles represent the voters.
    The graph on the top shows a neighborhood graph of five districts,
    each district consisting of two voters.
    The task is to dissolve three districts such that
    each remaining district contains five voters.
    The graphs in the middle show two possible realizations of dissolutions.
    The graphs on the bottom show the two corresponding outcomes.
    The arrows in the ``middle graphs'' point from the districts to be dissolved to the ``target
    districts'' and the white circle labels on the arrows
    represent the voters \moved along the arrows.}
\label{fig:dissol-ex}
\end{figure}

\begin{definition}[Dissolution]\label[definition]{def:dissolution}
  Let $G$ be an undirected graph,
  let $\Vdis \subset V(G)$ be a subset of \emph{districts to dissolve}
  and let $\z\colon \A{\Vdis,G} \to \{0,\ldots,\sold\}$ be a function that
  describes how many voters shall be \emph{\moved} from one district to its
  non-dissolved neighbors. Then, $(\Vdis, \z)$ is called an \emph{\dissol} for
  $G$ if
  \begin{enumerate}[a)]
    \item\label{prop:a} no voter remains in any dissolved district:
       \[\constFull\text{, and}\]
    \item\label{prop:b} the size of all remaining (non-dissolved) districts increases by $\Deltas$:
       \[\constEqualSize\text{.}\]
    \setcounter{saveenum}{\value{enumi}}
 \end{enumerate}

 \pagebreak[3]
 Throughout this work, we use
 \begin{itemize}
  \item $\snew := \sold + \Deltas$ to denote the new district size,
  \item $\kdis:=|\Vdis| = |V(G)| \cdot \Deltas /\snew$ to
        denote the number of dissolved districts, and
  \item $\krem:=|V(G)|-\kdis{}$ to denote the number of remaining, non-dissolved districts.
 \end{itemize}
\end{definition}

\noindent We write \emph{\dissolution} instead of $(\sold,
\Deltas)$-\dissolution when $\sold$ and $\Deltas$ are clear from the context.
 By definition, a dissolution ensures that the numbers
of voters moving between districts fulfill the given constraints on
the district sizes, that is, the size of each remaining district
increases by~$\Deltas$.
\cref{fig:dissol-ex} gives an example illustrating two possible \Sdissol{2}{3}{}s
for a $5$-vertex graph.

Motivated from social choice application scenarios, we additionally assume that
each voter supports one of two parties \partyA and \partyB.
We then search for a dissolution such that the number of remaining districts won by
party \partyA is maximized.  Here, a district is won by the party that
is supported by a strict majority of the voters inside the
district. This yields the notion of a \emph{biased dissolution}, which
is defined as follows:

\begin{definition}[Biased dissolution]\label[definition]{def:biased}
  Let $G$ be an undirected graph and
  let $\w\colon V(G)\to \{0,\ldots,\sold\}$ be an
  \emph{\partyAsupporter distribution}, where $\w(v)$~denotes the number
  of \partyAsupporters in district~$v\in V$.
  Let $(\Vdis, \z)$ be an \dissol for~$G$, that is,
  Properties~\ref{prop:a}) and~\ref{prop:b}) from \cref{def:dissolution} are fulfilled.
  Let $\kwin \in \mathbb{N}$ be the minimum number of districts 
  that party~\partyA shall win after the dissolution 
  and $\za \colon \A{\Vdis,G} \to \{0,\ldots,\sold\}$
  be an \emph{\partyA-supporter movement}, where $\za(v', v)$~denotes the number
  of \partyAsupporters \moving from district~$v'$ to district~$v$.
  Finally, let $\Vwin\subseteq V(G)\setminus \Vdis$ be a size-$\kwin$ subset of districts.
  Then, $(\Vdis, \z, \za, \Vwin)$ is called an \emph{\biaseddissol} for~$(G, \alpha)$ if
  \begin{enumerate}[a)]
    \setcounter{enumi}{\value{saveenum}}
    \item \label{prop:c}
       a district does not receive more \partyAsupporters from a dissolved
       district than the total number of voters it receives from that district:
       \[\constAsupportersUpperbound\text{,}\]
    \item \label{prop:d}
       no \partyAsupporters remain in any dissolved district:
       \[\constAsupportersFull\text{, and}\]
     \item \label{prop:e}
       each district in~$\Vwin$ has a strict majority of \partyAsupporters:
       \[\constWinningdistricts\text{.}\]
  \end{enumerate}
We also say that a district \emph{wins} if it has a strict majority of \partyAsupporters,
and \emph{loses} otherwise.
\end{definition}

\begin{figure}[t!]
  \centering
    \begin{tikzpicture}[scale=0.9]
      \begin{scope}
        \node at (2,3.6) (null) {};
        \nodeABB{v1}{below=0pt of null}
        \nodeABB{v3}{below=0.8 of v1}
        \nodeABB{v2}{left=2.0 of v3v1}
        \nodeABB{v4}{right=2.0 of v3v1}
        \nodeAAA{v5}{below=0.8 of v3}

        \node [left = 4cm of v1v1] (n) [transform shape]  {$n=5$};
        \node[below = 15pt of n.west,anchor=west] (sold) [transform shape]  {$\sold=3$};
        \node[below = 15pt of sold.west,anchor=west] (snew) [transform shape]  {$\Deltas=2$};
        \draw [ZEdge] (v1) -- (v2);
        \draw [ZEdge] (v1) -- (v3);
        \draw [ZEdge] (v1) -- (v4);
        \draw [ZEdge] (v2) -- (v5);
        \draw [ZEdge] (v3) -- (v5);
        \draw [ZEdge] (v4) -- (v5);
        
        \node [below=0.5cm of v5v1] (sepStart) [transform shape] {};
        \node [below=3.4cm of sepStart] (sepEnd) [transform shape] {};
         
        \draw [transform shape] (sepStart) -- (sepEnd);
      \end{scope}
      
      \begin{scope}[xshift=-3.3cm, yshift=-2.4cm]
        \node at (2,3.6) (null) {};
        \nodeABBd{v1}{below=0pt of null}
        \nodeABB{v3}{below=0.8 of v1}
        \nodeABB{v2}{left=2.0 of v3v1}
        \nodeABB{v4}{right=2.0 of v3v1}
        \nodeAAAd{v5}{below=0.8 of v3}
        \draw [ZArrow] (v1) -- (v2) node [midway]{\moveB};
        \draw [ZArrow] (v1) -- (v3) node [midway]{\moveA};
        \draw [ZArrow] (v1) -- (v4) node [midway]{\moveB};
        \draw [ZArrow] (v5) -- (v2) node [midway]{\moveA};
        \draw [ZArrow] (v5) -- (v3) node [midway]{\moveA};
        \draw [ZArrow] (v5) -- (v4) node [midway]{\moveA};
      \end{scope}
      
      \begin{scope}[xshift=3.3cm, yshift=-2.4cm]
        \node at (2,3.6) (null) {};
        \nodeABBd{v1}{below=0pt of null}
        \nodeABB{v3}{below=0.8 of v1}
        \nodeABB{v2}{left=2.0 of v3v1}
        \nodeABB{v4}{right=2.0 of v3v1}
        \nodeAAAd{v5}{below=0.8 of v3}
        \draw [ZEdge,dotted] (v1) -- (v4);
        \draw [ZEdge,dotted] (v2) -- (v5);
        \draw [ZArrow] (v1) -- (v2) node [midway]{\moveBB};
        \draw [ZArrow] (v1) -- (v3) node [midway]{\moveA};
        \draw [ZArrow] (v5) -- (v3) node [midway]{\moveA};
        \draw [ZArrow] (v5) -- (v4) node [midway]{\moveAA};
      \end{scope}

      \begin{scope}[xshift=-3.3cm, yshift=-6.3cm]
        \node at (2,3.6) (null) {};
        \nodeAAABB{v3}{below=0pt of null}
        \nodeAABBB{v2}{left=2.0 of v3v1}
        \nodeAABBB{v4}{right=2.0 of v3v1}
      \end{scope}
      
      \begin{scope}[xshift=3.3cm, yshift=-6.3cm]
        \node at (2,3.6) (null) {};
        \nodeAAABB{v3}{below=0pt of null}
        \nodeABBBB{v2}{left=2.0 of v3v1}
        \nodeAAABB{v4}{right=2.0 of v3v1}
      \end{scope}
  \end{tikzpicture}
  \caption{An illustration of a \AbiasedSdissol{1}{3}{2}
    (left) and a \AbiasedSdissol{2}{3}{2} (right).
    Black circles represent \partyAsupporters, while white circles
    represent \partyBsupporters.
    The graph on the top shows a neighborhood graph of five districts,
    each district consisting of three voters.
    The task is to dissolve two districts such that each remaining
    district contains five voters.
    The graphs in the middle show two possible realizations of dissolutions.
    The graphs on the bottom show the two corresponding outcomes.
    The arrows point from the districts to be dissolved to the ``target
    districts'' and the black/white circle labels on the arrows
    indicate which kind of voters are \moved along the arrows.}
\label{fig:biased-ex}
\end{figure}
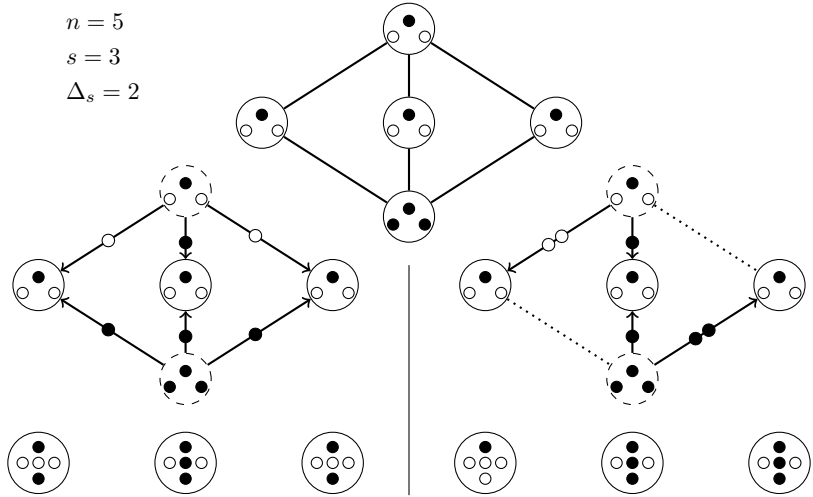

\noindent \cref{fig:biased-ex} shows two biased dissolutions: one with $\kwin=1$ and the other one with $\kwin=2$.  We are now ready to formally state the definitions of the two computational dissolution problems (in their decision versions) that we discuss in this work:

\decprob{\probValid}
{An undirected graph $G=(V,E)$ and positive integers $\sold$ and $\Deltas$.}
{Is there an \dissol for $G$?}

\decprob{\probWinner}
{An undirected graph $G=(V,E)$, positive integers~$\sold, \Deltas,
  \kwin$, and an \partyAsupporter distribution~$\w: V \rightarrow
  \{0,\dots,\sold\}$.}
{Is there an \biaseddissol for $(G, \alpha)$?}

\noindent
Note that \probValid is equivalent to \probWinner with $\kwin=0$.
As we will see later, \probValid and \probWinner are \NP-complete in general.
In this work, we additionally look into special cases 
and investigate what the causes of intractability may be.

\section{Relations to Established Models}
\label{sec:relations}
In this section, we identify relations of our model to established
graph concepts like matchings, flow networks, or star partitions.
This will also be useful for proofs in later sections.
In \cref{sec:Relation2Flownetworks}, we show that
\probValid and \probWinner instances where the roles of the districts are
already known can be translated into flow networks.
In \cref{sec:Relation2Starmatching} we show that \dissolution{}s generalize
star partitions and perfect matchings.

\subsection{Flow Networks}
\label{sec:Relation2Flownetworks}
Sometimes the districts that are to be dissolved and the districts
that are to be won are not arbitrary but already determined beforehand.
For this case we show that \probWinner can be modeled as a network
flow problem, which can be solved in polynomial time.

\newcommand{\criticalvertex}{\ensuremath{v}}
\newcommand{\criticalnode}{\ensuremath{u}}
\newcommand{\Sum}[3]{\ensuremath{\smashoperator{\sum_{x\in{#1(#2)}}}#3(x,#2)}}
\newcommand{\Suminl}[3]{\ensuremath{\sum_{x\in{#1(#2)}}#3(x,#2)}}

\begin{theorem}\label{thm:P-Vdis+Vwin-known}
Let $(G,\sold,\Deltas,\kwin,\w)$ be a \probWinner instance, and let 
$\Vdis,\Vwin\subset V(G)$ be two disjoint fixed subsets of districts.
The problem of deciding whether $(G,\w)$ admits an $\kwin$-biased $(\sold,\Deltas)$-dissolution 
$(\Vdis, \z, \za, \Vwin)$
can be reduced in linear time to a maximum flow problem with $2|V(G)|+2$ nodes, $2|V(G)|+3|E(G)|$~arcs,
and maximum arc capacity $\max(\sold,\Deltas)$. 
\end{theorem}

\begin{proof}
 Denote the set of remaining districts by $\Vrem$, that is, $\Vrem:=V(G) \setminus \Vdis$.
 With $\Vwin\subseteq \Vrem$ given beforehand,
 we can compute how many \partyAsupporters 
 a district~$v \in \Vwin$ needs from its neighboring dissolved districts
 in order to win after the dissolution. 
 With also $\Vdis$ given beforehand,
 we can use a flow network with two nodes corresponding to each district
 to compute an \biaseddissol{}.
 
  To this end, we first remove all edges between two vertices from $\Vdis$ or between
  two vertices from $\Vrem$ since only edges between $\Vdis$ and $\Vrem$ may be taken into
  account for the dissolution.
  Doing this, we obtain a bipartite neighborhood graph with the
  two disjoint vertex sets~$\Vdis=\{d_1,\dots,d_k\}$ and $\Vrem=\{r_1,\dots,r_{n-k}\}$. 
  Second, we observe that, in order to let a district~$r\in \Vrem$ win after the dissolution,
  $r$~needs at least $\max\{0, \lceil(\snew+1)/2\rceil - \w(r)\}$ additional \partyAsupporters. 
  Hence, we compute a ``demand'' function~$\demand: \Vrem \to \{0,\ldots,\lceil(\snew+1)/2\rceil\}$ for each non-dissolved district~$r$ by $\demand(r) := \max\{0, \lceil(\snew+1)/2\rceil - \w(r)\}$ if $r \in \Vwin$ and $\demand(r) := 0$ otherwise.
  
 The idea now is to construct a flow network which models the movement of \partyAsupporters
 that are necessary for a district in $\Vwin$ to win and models the movement of the remaining
 voters necessary to end up with district size~$\snew$ separately.
 More precisely, we split each $d \in \Vdis$ into a node~$d^{A}$, modeling the supply of
 \partyAsupporters from~$d$, and into a node~$d^{B}$, modeling the supply of the \partyBsupporters from~$d$.
 Similarly, we split each $r \in \Vrem$ into a node~$r^{A}$, modeling the demand for
 \partyAsupporters for $r$, and into a node~$r^{AB}$, modeling the remaining demand for voters,
 that is, voters to finally end up with district size~$\snew$.
 Now, following the constraints given by the neighborhood graph, \partyAsupporters may move
 in order to satisfy some demand for \partyAsupporters or in order to satisfy the general demand
 on voters.
 Clearly, \partyBsupporters may only move in order to satisfy the general demand on voters.
  
\tikzstyle{stnode}=[circle,fill=black!25,minimum size=17pt,inner sep=0pt]
\tikzstyle{disnode}=[circle,fill=black!25,minimum size=17pt,inner sep=0pt]
\tikzstyle{remnode}=[circle,fill=black!25,minimum size=17pt,inner sep=0pt]
\tikzstyle{arc}=[thick,->,shorten <= 0pt, shorten >= 0pt,sloped]
\tikzstyle{edge}=[thick,-]
\tikzstyle{edgecap}=[scale=.8, midway,fill=white]

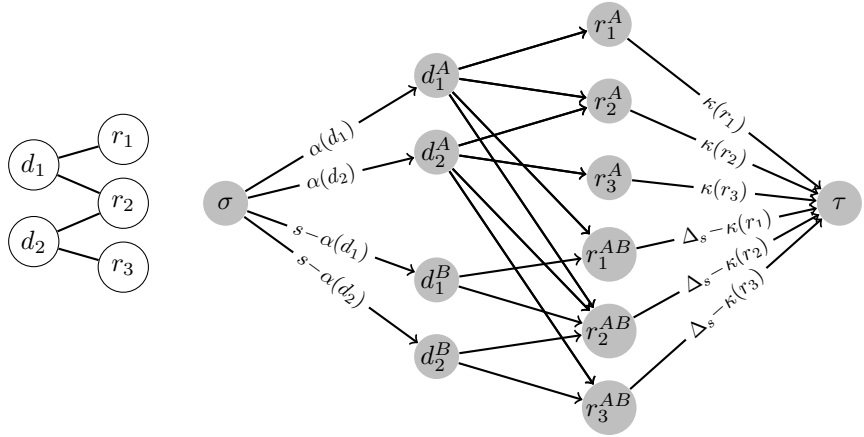
\begin{figure}[t!]
  \centering
 \begin{tikzpicture}[scale=0.85, node distance=\layersep]
  \tikzstyle{every pin edge}=[<-,shorten <=1pt]
  \def\layersep{3}
  \begin{scope}
    \tikzstyle{origdisnode}=[draw, circle, minimum size=19pt, inner sep=0pt]
    \tikzstyle{origremnode}=[draw, circle, minimum size=19pt, inner sep=0pt]
    \node[origdisnode] (origdis-1) at (0,0.6) {$d_1$};
    \node[origdisnode] (origdis-2) at (0,-0.6) {$d_2$};
    
    \node[origdisnode] (origrem-1) at (1.4,1) {$r_1$};
    \node[origdisnode] (origrem-2) at (1.4,0) {$r_2$};
    \node[origdisnode] (origrem-3) at (1.4,-1) {$r_3$};

    \draw[edge] (origdis-1) -- (origrem-1);
    \draw[edge] (origdis-1) -- (origrem-2);
    \draw[edge] (origdis-2) -- (origrem-2);
    \draw[edge] (origdis-2) -- (origrem-3);
  \end{scope}
  
  \begin{scope}[xshift=3cm]
  \node[stnode] (s) at (0,0) {$\networks$};
  \node[stnode] (t) at (\layersep*2+3.6,0) {$\networkt$};

  \foreach \i in {1,2}{
    \node[disnode] (vdisA-\i) at (3.3, -1.2*\i+3.2) {$d_{\i}^{A}$};
    \node[disnode] (vdis-\i) at (3.3, -1.2*\i+0.0) {$d_{\i}^{B}$};
  }

  \foreach \i/\j in {1/1,2/2,3/3}{
    \node[remnode] (vremA-\i) at (6.0,-1.2*\i+4.0) {$r_{\j}^{A}$};
    \node[remnode] (vrem-\i) at (6.0,-1.2*\i+0.4) {$r_{\j}^{AB}$};
  }

  \foreach \target in {1,2} {
    \draw[arc] (s) -- node[edgecap] {$\w(d_{\target})$} (vdisA-\target);
    \draw[arc] (s) -- node[edgecap] {$\sold\!-\!\w(d_{\target})$} (vdis-\target);
  }
  
  \foreach \source/\j in {1/1,2/2,3/3} {
    \draw[arc] (vremA-\source) -- node[edgecap] {$\demand(r_{\j})$} (t);
    \draw[arc] (vrem-\source) -- node[edgecap] {$\Deltas\!-\!\demand(r_{\j})$} (t);
  }
  
  \foreach \i/\j in {1/1,1/2,2/2,2/3} {
   \draw[arc] (vdisA-\i) -- (vremA-\j);
   \draw[arc] (vdisA-\i) -- (vrem-\j);
   \draw[arc] (vdisA-\i) -- (vremA-\j);
   \draw[arc] (vdisA-\i) -- (vrem-\j);
   \draw[arc] (vdis-\i) -- (vrem-\j);
  }
  \end{scope}
\end{tikzpicture}
\caption{An illustration of the network flow construction. 
Left: the graph~$G$ of an instance of \probWinner with $\Vdis=\{d_1,d_2\}$. 
Right: the corresponding network flow. The capacities of the arcs from dissolved nodes to non-dissolved nodes 
are omitted for the sake of brevity.}
\label{fig:flow network for known Vdis+Vwin}
\end{figure}

  Formally, we construct a flow network~$\networkI=\networkInst$ for our input instance
  $(G,\sold,\Deltas,\kwin,\w)$ as follows (see \cref{fig:flow network for known Vdis+Vwin}
  for an illustration).
  The node set~$\networkV$ in $\networkG$ consists of a source node~$\networks$, a target node~$\networkt$, two nodes $d_{i}^{A}$ and $d_{i}^{B}$ for each district~$d_{i} \in \Vdis$,
  and two nodes $r_{i}^{A}$ and $r_{i}^{AB}$ for each district~$r_{i} \in \Vrem$.
  In total, $\networkV$ has $2|V|+2$ nodes.

  The arcs in $\networkE$ are divided into three layers:
  \begin{enumerate}
   \item Arcs from the source node to all nodes corresponding to dissolved districts:
         For each dissolved district~$d_i \in \Vdis$, add to $\networkE$ two arcs~$(\networks,d_{i}^{A})$ and $(\networks, d_{i}^{B})$ with capacities~$\networkCap(\networks,d_{i}^{A})=\w(d_i)$ and $\networkCap(\networks,d_{i}^{B})= \sold - \w(d_i)$.
   \item Arcs from the nodes corresponding to dissolved districts
         to nodes corresponding to non-dissolved neighbored districts:
         For each dissolved district~$d_i \in \Vdis$ and for each $r_j\in N(d_i)$
         of its non-dissolved neighbors, add to $\networkE$ three arcs
         $(d_{i}^{A}, r_{j}^{A})$, $(d_{i}, r_{j}^{AB})$, and
         $(d_{i}^{B},r_{j}^{AB})$ with capacities
         $\networkCap(d_{i}^{A}, r_{j}^{A}) =
          \networkCap(d_{i}^{A}, r_{j}^{AB}) = \w(d_i)$
         and $\networkCap(d_{i}^{B},r_{j}^{AB})=\sold-\w(d_i)$.
   \item Arcs from all non-dissolved nodes to the target node:
         For each non-dissolved district~$r_{j}\in \Vrem$, 
         add to $\networkE$ two arcs~$(r_{j}^{A}, \networkt)$
         and $(r_{j}^{AB}, \networkt)$ with capacities
         $\networkCap(r_{j}^{A}, \networkt)=\demand(r_j)$
         and $\networkCap(r_{j}^{AB}, \networkt)=\Deltas-\demand(r_j)$.
  \end{enumerate} 
  This completes the description of the flow network construction. 

\begin{figure}[t!]
  \centering
\begin{tikzpicture}[scale=0.9]
  \begin{scope}
    \node[disnode] at (0,0) (vdisA) {$d_{i}^{A}$};
    \node[remnode, right = 2.1cm of vdisA] (vremA) {$r_{j}^{A}$};
    \node[remnode, below = 0.53 of vremA] (vrem) {$r_{j}^{AB}$};
    
    \draw[arc] (vdisA) -- node[edgecap] {$\za(d_i, r_j)$} (vremA);
    \draw[arc] (vdisA) -- node[edgecap] {$0$} (vrem);
  \end{scope}

  \begin{scope}[xshift=4.6cm]
    \node[disnode] at (0,0) (vdisA) {$d_{i}^{A}$};
    \node[remnode, right = 2.1cm of vdisA] (vremA) {$r_{j}^{A}$};
    \node[remnode, below = 0.53 of vremA] (vrem) {$r_{j}^{AB}$};
    
    \draw[arc] (vdisA) -- node[edgecap] {$\delta$} (vremA);
    \draw[arc] (vdisA) -- node[edgecap] {$\za(d_{i}, r_j)-\delta$} (vrem);
  \end{scope}
\end{tikzpicture}
\caption{Two cases of setting the flow values for arcs outgoing from $d_i^A$ nodes.
 Left: The sum of the flow through arcs to $r_{j}^A$ is at most $\demand(r_j)-\za(d_i, r_j)$.
 Right: The sum of the flow through arcs to $r_{j}^A$ is $\demand(r_j)-\delta$ where
 $0\le\delta<\za(d_i, r_j)$. 
}
\label{fig:flow values assignment}
\end{figure}
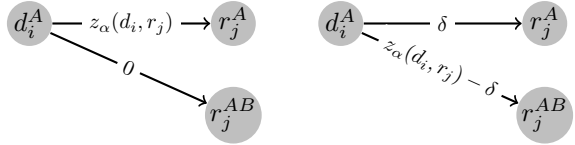

  We show that there is an \biaseddissol~$(\Vdis, \z, \za, \Vwin)$ for $(G,\w)$
  if and only if the constructed flow network $\networkI$ has a
  $(\networks,\networkt)$-flow of value~$\sold\cdot|\Vdis|$. 

  For the ``only if'' part, suppose that there is a dissolution $(\Vdis,\z,\za,\Vwin)$
  for~$(G,\w)$. 
  Construct a $(\networks,\networkt)$-flow~$f: \networkE \to \mathbb{R}$ 
  by defining $f(\networks,d_i^{A}) := \networkCap(\networks,d_i^{A})$
  and $f(\networks,d_i^{B}) := \networkCap(\networks,d_i^{B})$,
  where $d_i$ is a dissolved district.
  Then, define $f(r_j^{A},\networkt) := \networkCap(r_j^{A},\networkt)$
  and $f(r_j^{AB},\networkt) := \networkCap(r_j^{AB},\networkt)$,
  where $r_j$ is a non-dissolved district.
  Note that by definition of the network, this means that
  $f(\networks,d_i^{A})=\w(d_i)$ and $f(\networks,d_i^{B})=s-\w(d_i)$,
  where $d_i$ is a dissolved district, as well as that
  $f(r_j^{A},\networkt) = \demand(r_j)$ and $f(r_j^{AB},\networkt) = \Deltas - \demand(r_j)$,
  where $r_j$ is a non-dissolved district. 
  It remains to define the values of~$f$ for the arcs in layer~2.
  For each dissolved district $d_i \in \Vdis$ and for each~$r_j \in N(d_i)$
  of its non-dissolved neighbors, define
  $f(d_{i}^{B}, r_{j}^{AB}) := \z(d_i, r_j) - \za(d_i, r_j)$.
  To define also the flow values for arcs outgoing from a node~$d_i^{A}, 1 \le i \le k$, 
  we use the following procedure, where we remember in each step the total
  amount~$u(r_j^A)$ of flow going into $r_j^A$.
  We initialize~$u$ by setting~$u(r_j^A):=0$ for each~$r_j \in \Vrem$.
  Now, process all pairs $(d_i,r_j)$ with $d_i \in \Vdis$ and $r_j \in N(d_i)$
  in an arbitrary ordering, where the following two cases may occur
  (illustrated in \cref{fig:flow values assignment}).
  \begin{description}
   \item[Case 1.] 
    If $u(r_j^A)+\za(d_i,r_j)\le\demand(r_j)$, then
    increase $u(r_j^A)$ by $\za(d_i,r_j)$ and
    set $f(d_i^A,r_j^A):=\za(d_i,r_j)$ and $f(d_i^A,r_j^{AB}):=0$.
   \item[Case 2.] 
    If $u(r_j^A)+\delta=\demand(r_j)$ for some non-negative integer $\delta<\za(d_i,r_j)$, then
    increase $u(r_j^A)$ by~$\delta$ and
    set $f(d_i^A,r_j^A):=\delta$ and $f(d_i^A,r_j^{AB}):=\za(d_i,r_j)-\delta$.
  \end{description}
  Now, observe that by our definition of~$f$
  the flow value is $\sum_{(s,x)\in\networkE}{f(s,x)} = \sold \cdot |\Vdis|$.
  It remains to show that $f$ is valid.
  By our definition of~$f$, the flow value of each arc does not exceed its capacity.
  For each $d_i \in \Vdis$, the conservation property for the nodes $d_i^{A}$ and $d_i^{B}$
  is fulfilled by Property~\ref{prop:a}) (\cref{def:dissolution}) and~\ref{prop:d}) (\cref{def:biased}) of the \biaseddissolution.
  For each $r_j \in \Vrem$, the conservation property for the node $r_j^{A}$ is fulfilled by
  our definition of~$f$ (which ensures that the ingoing flow is at most $\demand(r_j)$) 
  and by Property~\ref{prop:e}) (\cref{def:biased}) of the \biaseddissolution
  (which ensures that the ingoing flow is at least $\demand(r_j)$).
  The conservation property for the node $r_j^{AB}$ is fulfilled by Property~\ref{prop:c})
  and~\ref{prop:e}) (\cref{def:biased}) of the \biaseddissolution (and the way we defined~$f$).
  
  For the ``if'' part, suppose that $f$ is a $(\networks,\networkt)$-flow for $\networkI$ with value $\sold \cdot |\Vdis|$. 
  Let $\za: \A{\Vdis,G} \to \{0,\ldots,\sold\}$ and $\z: \A{\Vdis,G} \to \{0,\ldots, \sold\}$ be two functions with values~$\za(d_i, r_j):=f(d_{i}^{A}, r_{j}^{A}) + f(d_{i}^{A},r_{j}^{AB})$ 
  and $\z(d_i,r_j):=\za(d_i,r_j)+f(d_{i}^{B},r_{j}^{AB})$. 
  One can verify that $(\Vdis, \z, \za, \Vwin)$ is an $\kwin$-biased $(\sold,\Deltas)$-dissolution for~$(G,\w)$ as follows:
  Property~\ref{prop:a}) (\cref{def:dissolution}) is fulfilled since the total flow going over $d_{i}^{A}$
  and $d_{i}^{B}$ has value exactly~$\sold$.
  Property~\ref{prop:b}) (\cref{def:dissolution}) is fulfilled since the total flow going over $r_{j}^{A}$
  and $r_{j}^{AB}$ has value exactly~$\Deltas$.
  Property~\ref{prop:c}) (\cref{def:biased}) is fulfilled by our definition of $\z$ and $\za$.
  Property~\ref{prop:d}) (\cref{def:biased}) is fulfilled since the total flow going over $d_{i}^{A}$
  is $\w(d_i)$.
  Property~\ref{prop:e}) (\cref{def:biased}) is fulfilled since the total flow going over $r_{j}^{A}$
  is $\demand(r_j)$.
\end{proof}

\noindent The following corollary shows that plain \dissolution{}s can be modeled using a much simpler
flow network in comparison to \biaseddissolution{}s.
In particular, all capacity values are either~$\sold$ or~$\Deltas$---a property
which will be important in later proofs.

\begin{corollary}
\label[corollary]{cor:flow}
Let $G$ be a graph and let $\Vdis\subset V(G)$ be a subset of vertices.
If there exists an \dissol~$(D,z)$ for $G$, then it can be found by computing the 
maximum flow in a network with $|V(G)|+2$~nodes and $|E(G)|+2|V(G)|$~arcs where all capacities 
are either $\sold$ or $\Deltas$.
\end{corollary}

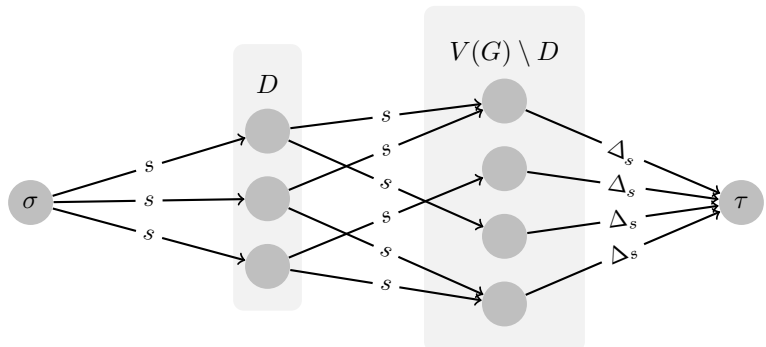
\begin{figure}[t]
 \centering
 \def\layersep{3.5}
 \begin{tikzpicture}[scale=0.9, node distance=\layersep]
  \tikzstyle{neuron}=[circle,fill=black!25,minimum size=17pt,inner sep=0pt]
  \tikzstyle{arc}=[thick,->,shorten <= 0pt, shorten >= 0pt,sloped]
  \tikzstyle{edge}=[thick,-]
  \tikzstyle{edgecap}=[scale=.9, midway,fill=white]

  \node[neuron] (s) at (0-\layersep,-1.5) {\networks};

  \node[yshift=0.5cm] (V-0) at (0,-0.3) {$\Vdis$};
  \foreach \y in {1,2,3}
   \node[yshift=0.5cm,neuron] (V-\y) at (0,-\y) {};%

  \node (U-0) at (\layersep,0.7) {$V(G) \setminus \Vdis$};
  \foreach \name / \y in {1,2,3,4}
   \node[neuron] (U-\name) at (\layersep,1-\y) {};%

  \node[neuron] (t) at (2*\layersep,-1.5) {\networkt};

  \foreach \source in {1,2,3}   
   \draw[arc] (s) -- node[edgecap] {$\sold$} (V-\source);
  \foreach \source / \dest in {1/1,1/3,2/1,3/2,3/4,2/4}
   \draw[arc] (V-\source) -- node[edgecap] {$\sold$} (U-\dest);
  \foreach \source in {1,2,3,4}
   \draw[arc] (U-\source) -- node[edgecap] {$\Deltas$} (t);

  \begin{pgfonlayer}{background}
   \node (V_fit) [draw=none,rounded corners,fill=gray!10,fit=(V-0) (V-3)] {};
   \node (U_fit) [draw=none,rounded corners,fill=gray!10,fit=(U-0) (U-4)] {};
  \end{pgfonlayer}
 \end{tikzpicture}
 \caption{Flow network for \probValid when the set~$\Vdis$ of districts to dissolve is known.}
\label{fig:simpleFlowExample}%
\end{figure}

\begin{proof}
 If the districts to dissolve are known and we only search for a \dissolution
 (in other words, $\kwin=0$),
 then the flow network used to compute a \dissolution from the proof of \cref{thm:P-Vdis+Vwin-known}
 basically reduces to a much simpler flow network.
 For this case, we can assume that $\Vwin=\emptyset$ and $\w(v)=0$ for all $v \in V(G)$,
 remove all arcs with capacity zero, and finally
 also remove nodes without a directed path from the source or to the sink.
 
 Doing this, we end up with the following:
 We have a source~$\networks$ and a sink~$\networkt$ and two additional layers of nodes:
 the first layer contains one node for each vertex from~$\Vdis$ and
 the second layer contains one node for each vertex from~$V(G) \setminus \Vdis$.
 There is an arc from the source~$\networks$ to each node in the first layer with capacity~$\sold$ and
 an arc from each node in the second layer to the sink~$\networkt$ with capacity~$\Deltas$.
 Finally, there is an arc of capacity~$\sold$ from a node in the first layer to a node in the second layer
 if and only if the corresponding vertices in the neighborhood graph~$G$ are adjacent.
 See \cref{fig:simpleFlowExample} for an illustration.
\end{proof}

\noindent Contrasting the \pt{} solvability when $\Vdis$ and $\Vwin$ are known,
we obtain \NP-completeness for \probWinner once at least one of the two
sets~$\Vdis$ and~$\Vwin$ is unknown.
\probValid is the special case of \probWinner with $\Vwin=\emptyset$
and we will see in \cref{sec:dissolution-dichotomy} that \probValid
is \NP-complete for the case
of $\sold \neq \Deltas$ (\cref{thm:dissolution-dichotomy}).
This means that \probWinner is \NP-hard even if the set~$\Vwin$ is known to be empty. 
For the case that only the set~$\Vdis$ of dissolved districts is given
beforehand, it remains to decide how many \partyAsupporters are
\moved to a certain non-dissolved district.
We will see in \cref{sec:biaseddissolution-dichotomy}, however, that in the hardness
construction for \cref{thm:biaseddissolution-dichotomy}
it is already fixed which districts are to be dissolved.
This means that \probWinner is \NP-hard even if the set~$\Vdis$ of dissolved districts is given
beforehand.
Summarizing, \probWinner is \NP-hard even if either the set~$\Vdis$ of districts
to dissolve or the set~$\Vwin$ of districts to win is known.

With the help of the above flow network construction from \cref{thm:P-Vdis+Vwin-known},
we can design an exact algorithm for \probWinner that runs in polynomial time when the number of districts is a constant.
Since the degree of the polynomial does not depend on the number of districts, this
means fixed-parameter tractability with respect to the number of districts
(see \cite{DF13,FG06,Nie06} for details on fixed-parameter tractability).

\begin{corollary}
\label[corollary]{cor:exact-algorithm}
Any instance $(G,\sold, \Deltas,\w)$ of \probWinner can be solved in 
$O(3^{|V(G)|} \cdot (\max(\sold, \Deltas) \cdot |V(G)|\cdot|E(G)| + |V(G)|^3) )$ time.
\end{corollary}

\begin{proof}
Since each district will either be dissolved, won, or lost, there are at most 
$3^{|V(G)|}$ different ways to fix the roles of all $|V(G)|$~districts.
In each case, we can construct a flow network with $O(|V(G)|)$ nodes 
and maximum capacity~$\max(\sold,\Deltas)$ in 
$O(\max(\sold,\Deltas)\cdot|V(G)|\cdot|E(G)|)$ time and compute the maximum flow 
(\cref{thm:P-Vdis+Vwin-known}) to solve \probWinner.  
Hence, by using an $O(|V(G)|^3)$-time maximum flow algorithm we solve 
\probWinner in $O(3^{|V(G)|}(\max(\sold,\Deltas)\cdot|V(G)|\cdot|E(G)|+|V(G)|^3))$ time.
\end{proof}

\subsection{Relation to Star Partition and Matching}
\label{sec:Relation2Starmatching}
In this subsection, we analyze how \dissolution{}s relate to star partitioning
and matching.
If $\Deltas=1$, then each non-dissolved district receives exactly one additional
voter from one of its neighboring districts.
Each dissolved district has to move exactly one voter to each of $\sold$~neighboring districts.
Hence, it is easy to see that a graph has an \Sdissol{s}{1} if and only
if it has an \starpartition{s}.

Using the flow construction from \cref{cor:flow}, we can even show that this equivalence
to star partition generalizes to the case that~$\sold$ is an integer multiple of any~$\Deltas$.

\newcommand{\stardissol}[1]{\Sdissol{#1\cdot\Deltas}{\Deltas}}
\begin{proposition}
\label[proposition]{prop:starpacking-relation}
There exists a \stardissol{t} for an undirected graph~$G$ if and only if~$G$~has a \starpartition{t}.
\end{proposition}

\begin{proof}
  If $G=(V,E)$ can be partitioned into $t$-stars, then it is easy to see that
  there is a \stardissol{t} for $G$:
  Let $C=\{c_1,\dots,c_\kdis\} \subset V$ be the set of $t$-star centers and
  let $L_i \subset V, 1 \le i \le \kdis,$  be the set of leaves of the $i$-th star.
  Define function~$\z: \A{C,G} \to \{0,\dots,t\cdot\Deltas\}$
  so that, for all $(c_i,l)\in \A{C,G}$, $\z(c_i,l) := \Deltas$ if $l \in L_i$ and $\z(c_i,l) := 0$ otherwise.
  Obviously, $(C, \z)$ is a \stardissol{t} for~$G$.

  Now, let $(\Vdis,z)$ be a \stardissol{t} for~$G$.
  We show that $G$ can be partitioned into $t$-stars with $\Vdis$ being the $t$-star centers.
  To this end, consider the network flow constructed in \cref{cor:flow}
  and modify the network as follows.
  For each arc, divide its capacity by~$\Deltas$.
  Clearly, if there is a flow with value~$|\Vdis|\cdot t\cdot\Deltas=|V \setminus \Vdis| \cdot \Deltas$,
  then the modified network has a flow with value~$|\Vdis|\cdot t=|V \setminus \Vdis|$.
  As all capacities are integers, there exists a maximum flow~$f$ such that for each arc~$a$
  it holds that $f(a)$ is integer~\cite{AMO93}.
  Hence, a partition of~$G$ into $t$-stars consists of one star for each~$v_i \in \Vdis$
  such that $v_i$~is the star center connected to its leaves~$L_i=\{u \mid f(v_i,u)=1\}$.
\end{proof}

\noindent Since a \starpartition{t} with $t=1$ is a perfect matching, we obtain the following corollary.

\begin{corollary}
\label[corollary]{cor:dissolution-P}
There exists an \Sdissol{$\sold$}{$\sold$} for an undirected graph~$G$ if and only if~$G$~has a perfect matching.
\end{corollary}

\section{Complexity dichotomy with respect to district sizes}
\label{sec:Dichotomy}
In this section, we study the computational complexity of \probValid and \probWinner
with respect to the relation of the district size~$\sold$ to the district size increase~$\Deltas$.
We show that \probValid is \pt{}
solvable if $\sold=\Deltas$, and \NP-complete
otherwise~(\cref{thm:dissolution-dichotomy}).
\probWinner is \pt{} solvable if $\sold=\Deltas=1$,
and \NP-complete otherwise~(\cref{thm:biaseddissolution-dichotomy}).

We start by showing a useful structural observation for \dissolution{}s.
More precisely, we observe a symmetry concerning the district
size~$\sold$ and the district size increase~$\Deltas$ in the sense
that exchanging their values yields an equivalent instance of
\probValid.
Intuitively, the idea behind the following lemma is that the roles of
dissolved and non-dissolved districts in a given~\dissol can in fact
be exchanged by ``reversing'' the movement of voters to obtain a~\Sdissol{\Deltas}{\sold}.

\begin{lemma}
\label[lemma]{lem:d-delta-mirror}
There exists an \dissol for an undirected graph~$G$ if and only if there 
exists a \Sdissol{\Deltas}{\sold} for~$G$.
\end{lemma}

\begin{proof}
 Let $(\Vdis,z)$ be an \dissol for $G$.
 We show that $(V(G)\setminus\Vdis,z')$, where $z'$ is defined by $z'(x,y):=z(y,x)$
 is a \Sdissol{\Deltas}{\sold} for~$G$:
 First, observe that the domain of $z'$ is correct:
 \begin{align*}
  Z(V(G)\setminus\Vdis,G) =& \{(x,y) \mid x\in V(G)\setminus\Vdis \wedge y \in V(G)\setminus(V(G)\setminus\Vdis)\\ &\wedge \{x,y\} \in E(G))\}\\
  =&\{(x,y) \mid x\in V(G)\setminus\Vdis \wedge y \in \Vdis \wedge \{x,y\} \in E(G))\}.
 \end{align*}
 Second, observe that $(V(G)\setminus\Vdis,z')$ fulfills all properties of \cref{def:dissolution}:
 Property~\ref{prop:a}) is fulfilled for~$(V(G)\setminus\Vdis,z')$ if and only if
 Property~\ref{prop:b}) is fulfilled for~$(\Vdis,z)$, and Property~\ref{prop:b})
 is fulfilled for~$(V(G)\setminus\Vdis,z')$ if and only if Property~\ref{prop:a}) is fulfilled
 for~$(\Vdis,z)$.
\end{proof}

\subsection{Dissolution}
\label{sec:dissolution-dichotomy}

In this subsection, we show a \classP vs.\ \NP dichotomy %
of \probValid with respect to the district size~$\sold$ and the size increase~$\Deltas$.
Observe that from \cref{cor:dissolution-P} it directly follows that \probValid is
\pt{} solvable if $\sold=\Deltas$.

If $\sold\neq \Deltas$, then \probValid is \NP-complete. 
We can use a result from number theory to encode instances of the \NP-complete
\probExCover problem into instances of \probValid.

\decprob{\probExCover}
{A finite set $X$ and a collection $\mathcal{C}$ of subsets of $X$ of
  size $t$.}
{Is there a subcollection $\mathcal{C}' \subseteq \mathcal{C}$ that
  partitions $X$, that is, each element of $X$ is contained in exactly one subset in $\mathcal{C}'$?}
\noindent Now, let us briefly recall some elementary number theory.

\begin{lemma}[B\'ezout's identity]\label[lemma]{lem:bezout's identity}
  Let $a$ and $b$ be two positive integers and let $g$ be their greatest common divisor.
  Then, there exist two integers~$x$ and $y$ with $ax + by = g$.
\end{lemma}

\noindent Moreover, $x$ and $y$ in \cref{lem:bezout's identity}
can be computed in polynomial time using the
extended Euclidean algorithm~\citep[Section~31.2]{CLRS09}.  Indeed, we
can infer from Lemma~\ref{lem:bezout's identity} 
that any two integers $x'$ and $y'$
with $x'=ix+jb/g$ and $y'=iy-ja/g$ for some $i,j \in \mathbb{Z}$
satisfy $ax' + by' = ig$.  We will make use of this fact several times
in the \NP-hardness proof of the following theorem.

\begin{theorem}
\label{thm:dissolution-dichotomy}
If $\sold=\Deltas$, then \probValid is solvable in $O(n^{\omega})$ time
(where $\omega$ is the matrix multiplication exponent); otherwise the
problem is \NP-complete.
\end{theorem}

  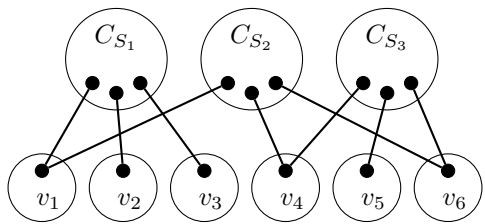
\begin{figure}[t]
    \centering
    \begin{tikzpicture}[scale=0.9]
      \tikzstyle{node}=[circle, draw, fill=black, inner sep=0pt, minimum size=5pt]
      \foreach \i in {1,2,...,6} {
        \coordinate (center) at (.5+1.2*\i+1.2,0);
        \draw (center) circle(.5);
        \draw (center) ++ (0,.25) node[node] (v\i) {};
        \draw (center) ++ (.1,-.2) node {$v_\i$};
      }
      \foreach \i in {1,2,3} {
        \coordinate (center) at (2*\i+2,1.9);
        \draw (center) circle(.75);
        \draw (center) ++ (225:.5) node[node] (S\i1) {};
        \draw (center) ++ (270:.5) node[node] (S\i2) {};
        \draw (center) ++ (315:.5) node[node] (S\i3) {};
        \node[above] at (center) {$C_{S_\i}$};
      }
      \draw[ZEdge] (v1) -- (S11);
      \draw[ZEdge] (v2) -- (S12);
      \draw[ZEdge] (v3) -- (S13);
      \draw[ZEdge] (v1) -- (S21);
      \draw[ZEdge] (v4) -- (S22);
      \draw[ZEdge] (v6) -- (S23);
      \draw[ZEdge] (v4) -- (S31);
      \draw[ZEdge] (v5) -- (S32);
      \draw[ZEdge] (v6) -- (S33);
    \end{tikzpicture}
    \caption{The constructed instance for $t=3$.}
    \label{fig:xsc-reduction-instance}
  \end{figure}

\begin{proof}
  First, \cref{cor:dissolution-P} says that there is an~\Sdissol{s}{s}
  if and only if there is a perfect matching in~$G$, which can be computed in
  $O(n^{\omega})$ time with $\omega$ being the smallest exponent such that
  matrix multiplication can be computed in $O(n^{\omega})$ time.
  Currently, the smallest known upper bound of $\omega$ is 2.3727~\cite{Wil12}.
  
  \newcommand{\subsetS}{\ensuremath{S}}
  For the case~$\sold \neq \Deltas$,
  we show that \probValid is \NP-complete if $\sold > \Deltas$.
  Due to \cref{lem:d-delta-mirror}, this also transfers to the cases where $\sold < \Deltas$.
  First, given a \probValid instance~$(G,\sold,\Deltas)$ and a function~$\z: \A{\Vdis,G} \rightarrow
  \{0,\dots,\sold\}$ where $\Vdis \subset V(G)$, 
  one can check in polynomial time whether $(\Vdis, \z)$ is an \dissol.
  Thus, \probValid is in~\NP.

  To show the \NP-hardness result, 
  we give a reduction from the \NP-complete \probExCover~\cite{GJ79}
  for $t := (\sold+\Deltas)/g > 2$, where $g := \gcd(\sold, \Deltas)$
  $\leq \Deltas$ is the
  greatest common divisor of~$\sold$ and~$\Deltas$.
  
  Given an \probExCover instance $(X,\mathcal{C})$,
  we construct a \probValid instance $(G,\sold,\Deltas)$ with a
  neighborhood graph $G = (V, E)$ defined as follows:
  For each element~$u \in X$, 
  add a clique $C_u$ of properly chosen size~$q$ to~$G$ and 
  let~$v_u$ denote an arbitrary fixed vertex in~$C_u$.
  For each subset~$\subsetS \in \mathcal{C}$, 
  add a clique $C_{\subsetS}$ of properly chosen size $r \geq t$ to~$G$ 
  and connect each $v_u$ for $u \in {\subsetS}$ to a unique vertex in $C_{\subsetS}$.
  \cref{fig:xsc-reduction-instance} shows an example of the
  constructed neighborhood graph for $t=3$.

  Next, we explain how to choose the values of~$q$ and~$r$.
  We set $q = x_q + y_q$, where $x_q \geq 0$ and $y_q \geq 0$ are
  integers satisfying $x_q\sold - y_q\Deltas = g$.
  Such integers exist by \cref{lem:bezout's identity}.
  The intuition behind this is as follows: 
  Dissolving~$x_q$ districts
  in $C_u$ and moving the voters to~$y_q$ districts in~$C_u$ creates
  an overflow of exactly~$g$ voters, who have to \move out of $C_u$.
  Note that the only way to \move voters into or out of $C_u$ is via
  district~$v_u$.
  Moreover, 
  if the constructed instance~$(G,\sold,\Deltas)$ admits a dissolution,
  then exactly~$x_q$ districts in~$C_u$ are
  dissolved because dissolving more districts leads to an overflow of
  at least~$g + \sold + \Deltas > \sold$~voters,
  which is more than~$v_u$ can \move, whereas dissolving less districts
  yields a demand of at least~$\sold + \Deltas - g > \Deltas$~voters, which is more than~$v_u$ can receive.
  Thus, the district~$v_u$ must be dissolved since there is an overflow of $g$~voters to \move out of~$C_u$ and this can only be done via district~$v_u$.
  
  The value of $r \geq t$ is chosen in such a way that, for each subset~${\subsetS} \in
  \mathcal{C}$ and each element~$u \in {\subsetS}$, it is possible to \move~$g$ voters
  from $v_u$ to $C_{\subsetS}$ (recall that $v_u$ must be dissolved). 
  In other words, we require~$C_{\subsetS}$ to be able
  to receive in total $t \cdot g = \sold+\Deltas$ voters
  in at least $t$ non-dissolved districts.
  Thus, we set $r := x_r + y_r$, where $x_r \geq 0$ and $y_r \geq t$
  are integers satisfying $x_r\sold - y_r\Deltas = -(\sold+\Deltas)$.
  Again, since $-(\sold+\Deltas)$ is divisible by~$g$, such integers exist by
  \cref{lem:bezout's identity}.
  It is thus possible to dissolve~$x_r$ districts in~$C_\subsetS$
  \moving the voters to the remaining $y_r$~districts in~$C_\subsetS$ such that
  we end up with a demand of~$\sold + \Deltas$ voters in~$C_\subsetS$.
  Note that the only other possibility is to dissolve $x_r +
  1$ districts in~$C_{\subsetS}$ in order to end up with a demand
  of zero voters.
  In this case, no voters of any other districts connected to~$C_{\subsetS}$ 
  can \move to~$C_{\subsetS}$.
  By the construction of $C_u$, it is clear that it is also
  not possible to \move any voters out of~$C_{\subsetS}$ because no~$v_u$ can
  receive voters in any dissolution.
  Thus,  if the constructed instance~$(G,\sold,\Deltas)$ admits a dissolution,
  then either all or none of
  the districts~$v_u$ connected to some $C_{\subsetS}$ \move $g$ voters to $C_{\subsetS}$.

  We are now ready to show that $G$ has a \dissol if and only if $(X,\mathcal{C})$
  is a yes-instance of \probExCover.

  For the ``only if'' part, suppose that $(X,\mathcal{C})$ is a yes-instance, that is,
  there exists a partition $\mathcal{C}' \subseteq \mathcal{C}$ of~$X$.
  We can thus dissolve~$x_q$ districts in each~$C_u$ (including~$v_u$)
  and \move the voters such that all~$y_q$ non-dissolved
  districts receive exactly $\Deltas$ voters. This is always possible
  since~$C_u$ is a clique. If we do so, then, by construction, 
  $g$ voters have to \move out of each~$v_u$.
  Since~$\mathcal{C}'$ partitions $X$, each~$u \in X$ is contained in
  exactly one subset $S\in\mathcal{C}'$. We can thus \move the $g$
  voters from each~$v_u$ to~$C_{\subsetS}$.
  Now, for each~${\subsetS} \in \mathcal{C}'$, we dissolve
  any $x_r$ districts that are not adjacent to any~$v_u$ and for the
  subsets in $\mathcal{C} \setminus \mathcal{C}'$, we simply dissolve
  $x_r + 1$~arbitrary districts in the corresponding cliques.
  As already discussed, each~$C_{\subsetS}$ with $x_r$~dissolved districts
  receives $t \cdot g$ voters and each~$C_{\subsetS}$ with $x_r+1$~dissolved
  districts receives no voter.
  Thus, this in fact yields an~\dissol.

  For the ``if'' part, assume that there exists an \dissol for~$G$.  
  As already discussed, every \dissol generates an overflow of $g$
  voters in each $C_u$ that has to be \moved over $v_u$ to
  some district in~$C_{\subsetS}$.
  Furthermore, each $C_{\subsetS}$ either receives $g$ voters from all
  its adjacent $v_u$ or no voters at all. Therefore, the subsets ${\subsetS}$
  corresponding to cliques $C_{\subsetS}$ that receive $t \cdot g$ voters
  form a partition of $X$, showing that $(X,\mathcal{C})$ is a yes-instance.
\end{proof}

\subsection{Biased Dissolution}
\label{sec:biaseddissolution-dichotomy}
Since \probValid is a special case of \probWinner, the \NP-hardness results for
$\sold\neq\Deltas$ transfer to \probWinner.
It remains to see whether \probWinner remains \pt{} solvable
when $\sold=\Deltas$.
Interestingly, this is true for $\sold=\Deltas=1$, but \probWinner turns
\NP-hard when $\sold=\Deltas\ge2$.

To analyze the structure of \dissolution{}s, we introduce the concept of the
``edge set used by a dissolution'' which we will use in several proofs.
Let~$(\Vdis,\z)$ be a \dissolution of a graph~$G$.
Let $E_\z\subseteq E(G)$ contain all edges~$\{x,y\}$ with $(x,y)\in \A{\Vdis,G}$ and $\z(x,y)>0$. %
Then, we call $E_\z$ the \emph{edge set used by} the dissolution~$(\Vdis,\z)$.

The following lemma shows that finding an \biasedSdissol{1}{1} essentially corresponds to
finding a maximum-weight perfect matching.

\begin{lemma}
\label[lemma]{lem:biased 1-1-dissolution}
Let $(G=(V,E),\sold=1,\Deltas=1,\kwin,\w)$ be a \probWinner instance.
There is an \biasedSdissol{1}{1} for $(G,\w)$ if and only if there is a perfect 
matching of weight at least~$\kwin$ in $(G,w)$ 
with $w(\{x,y\}):= 1 $ if $\w(x)=\w(y)=1$, and $w(\{x,y\}):=0$ otherwise.
\end{lemma}
\begin{proof}
For the ``only if'' part, let $(\Vdis,\z,\za,\Vwin)$ be an \biasedSdissol{1}{1} for~$(G,\w)$.
Then, the edge set~$E_z\subseteq E$ used by $(\Vdis,\z,\za,\Vwin)$ partitions $G$ 
into $1$-stars or in other words, $E_z$ is a perfect matching for~$G$ (see \cref{prop:starpacking-relation}).
Note that a non-dissolved district can only win if it already contains an \partyAsupporter
and receives one additional \partyAsupporter.
By the construction of~$w$, this implies that the weight of each edge that connects a 
winning district is one (i.e., for each $e \in E_z$ it holds that $e \cap \Vwin\neq \emptyset$
if and only if $w(e)=1$).
Since $|\Vwin|\ge\kwin$, the perfect matching~$E_z$ has weight at least $\kwin$. 

For the ``if'' part, let $E'\subseteq E$ be a perfect matching of weight at least~$\kwin$.
By the construction of~$w$, $E'$ contains at least $\kwin$ edges, each of which has weight one.
 Then, we construct an \biasedSdissol{1}{1} $(\Vdis,\z,\za,\Vwin)$ as follows.
 For each edge~$\{x,y\} \in E'$, arbitrarily add one of its endpoints,
 say~$x$, to $\Vdis$ and set $\z(x,y):=1$.
 Furthermore, if $\w(x)=1$, then set $\za(x,y):=1$.
 If $w(\{x,y\})=1$, meaning that the districts corresponding to $x$ and $y$ have an \partyAsupporter each, then add~$y$ to $\Vwin$ since $y$~wins after the dissolution.
 Finally, $|\Vwin|\ge \kwin$ since $|E'|\ge \kwin$.
\end{proof}

\noindent As we have already seen from \cref{cor:dissolution-P}, the edge set used by a \Sdissol{1}{1} is a perfect matching.
This is useful to find a \pt{} algorithm solving \probWinner, exploiting that
maximum-weight perfect matchings can be computed in polynomial time.
Can we find similar useful characterizations for \AbiasedSdissol{\kwin}{\sold}{\sold}{}s for $\sold>1$?

Already for \Sdissol{2}{2}{}s, a characterization by the edge set used
is not as compact as for \Sdissol{1}{1}{}s:
The edge set used by a \Sdissol{2}{2} for some graph~$G$
corresponds to a partition of the graph into disjoint cycles of even length
and disjoint paths on two vertices.
For the case of \biasedSdissol{2}{2}{}, one would at least need some weights and it is
not clear how to find such a partition efficiently.
However, by appropriately setting~$\w$ and~$\kwin$, we can enforce that the edge
set used by any \biasedSdissol{2}{2} only induces cycles of lengths divisible by four:
We let each district have one \partyAsupporter and one \partyBsupporter (i.e., $\w: V\to \{1\}$ for each district~$v$) and 
let $\kwin:=|V(G)|/4$. 
Doing this we end up with a restricted two-factor problem which was already studied
in the literature~\cite{HKKK88}:

\decprob{\probRestTwoFactor}
  {An undirected graph $G=(V,E)$.}
  {Is there a two-factor~$E' \subseteq E$ such that the number of vertices
   in each connected component in~$(V,E')$ belongs to~$L$?}

\noindent A \emph{two-factor} of a graph~$G=(V,E)$ is a subset of edges $E' \subseteq E$ such that
each vertex in the subgraph $G':=(V,E')$ has degree exactly two, that is, $G'$~only
contains disjoint cycles.

\begin{lemma}
\label[lemma]{lem:biased 2-2-dissolution}
Let $G=(V,E)$~be an undirected graph with $4q$ vertices ($q\in\mathbb{N}$).
Then, $G$ has a two-factor~$E'$ whose cycle lengths are all multiples of four 
if and only if $(G,\w)$ admits a \AbiasedSdissol{q}{2}{2} with $\w(v)=1$ for all $v\in V$.
\end{lemma}

\begin{figure}
   \centering
    \begin{tikzpicture}[scale=0.9]
      \begin{scope}[xshift=-3.0cm, yshift=0cm]
        \nodeABd{v1}{}
        \nodeAB{v2}{right=1.5 of v1v1}

        \draw [ZArrow] (v1) -- (v2) node [midway,name=e12]{\moveAB};;

        \node [below=2pt of e12] (comment) {a length-two path};
      \end{scope}
      
      \begin{scope}[xshift=-3cm, yshift=2.7cm]
        \nodeAB{v1}{}
        \nodeABd{v2}{right=1.5 of v1v1}
        \nodeAB{v3}{below=1.1 of v2v1}
        \nodeABd{v4}{below=1.1 of v1v1}

        \draw [ZArrow] (v2) -- (v1) node [midway]{\moveA};;
        \draw [ZArrow] (v2) -- (v3) node [midway]{\moveB};;
        \draw [ZArrow] (v4) -- (v1) node [midway]{\moveA};;
        \draw [ZArrow] (v4) -- (v3) node [midway,name=e43]{\moveB};;

        \node [below=2pt of e43] (comment) {a length-four cycle};
      \end{scope}

      \begin{scope}[xshift=3cm, yshift=2.7cm]
        \nodeAB{v1}{}
        \nodeABd{v2}{right=1.5 of v1v1}
        \nodeAB{v3}{below right=1.1 of v2v1}
        \nodeABd{v4}{below left=1.1 of v3v1}
        \nodeAB{v5}{left=1.5 of v4v1}
        \nodeABd{v6}{below left=1.1 of v1v1}

        \draw [ZArrow] (v2) -- (v1) node [midway]{\moveA};;
        \draw [ZArrow] (v2) -- (v3) node [midway]{\moveB};;
        \draw [ZArrow] (v4) -- (v5) node [midway,name=e45]{\moveA};;
        \draw [ZArrow] (v4) -- (v3) node [midway]{\moveB};;
        \draw [ZArrow] (v6) -- (v1) node [midway]{\moveA};;
        \draw [ZArrow] (v6) -- (v5) node [midway]{\moveB};;

        \node [below=2pt of e45] (comment) {a length-six cycle};
      \end{scope}
      
  \end{tikzpicture}
  \caption{Graphs induced by edge sets used by an \AbiasedSdissol{\kwin}{2}{2}.
           In order to have a majority of A-supporters (black dots) in at least
           half of the new districts, each component must be a cycle of length
           divisible by four.}
 \label{fig:RestrTwoFactor}
\end{figure}
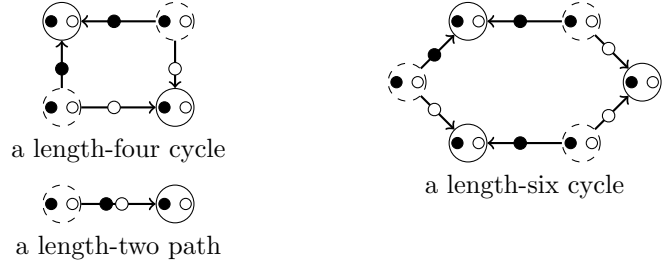

\begin{proof}
 \newcommand{\nqf}{\ensuremath{q}}
   For the ``only if'' part, let $E' \subseteq E$ be an edge subset such that each vertex in $G':=(V,E')$
   has degree two and $G'$~consists of disjoint cycles of lengths divisible by four.
   We now construct a \AbiasedSdissol{\nqf}{2}{2} $(\Vdis,\z,\za,\Vwin)$
   for $(G,\w)$.
   To this end, we start with $\Vdis:=\emptyset$, $\Vwin:=\emptyset$, %
   and do the following for each cycle~$c_1c_2\dots c_{4l}c_{1}, l\ge 1$.
   For each number~$i$ with $1 \le i \le 2l$, add $c_{2i}$ to $\Vdis$,
   and set $\z(c_{2i},c_{2i-1}):=\z(c_{2i},c_{(2i+1)\mod 4l}):=1$.
   For each $1\le i \le l$, we set
   \begin{alignat*}{2}
     \za(c_{4i-2},c_{4i-3})&:=1, &\quad \za(c_{4i-2},c_{4i-1})&:=0,\\
     \za(c_{4i},c_{(4i+1)\mod 4l})&:=1, & \za(c_{4i},c_{4i-1})&:=0.
   \end{alignat*}
   Doing this, every fourth vertex in each cycle receives two additional \partyAsupporters
   (see \cref{fig:RestrTwoFactor} for an illustration of the corresponding dissolutions).
   It is easy to verify that $(\Vdis,\z,\za,\Vwin)$ is indeed a \AbiasedSdissol{\nqf}{2}{2}.

   For the ``if'' part, let $(\Vdis,\z,\za,\Vwin)$ be a \AbiasedSdissol{\nqf}{2}{2} for $(G,\alpha)$.
   Furthermore, let $E_\z$ denote the edge set used by $(\Vdis,\z,\za,\Vwin)$.
   Each component~$C$ in $G[E_\z]$ 
   is either a path of length two or a cycle of even length and 
   consists of exactly $|V(C)|/2$ dissolved and $|V(C)|/2$ non-dissolved districts.
   Since each non-dissolved district needs at least two \partyAsupporters in order to win 
   and only $|V(C)|/2$ \partyAsupporters can be moved from the $|V(C)|/2$ dissolved districts,
   at most $|V(C)|/4$~districts can win.
   With $\kwin=q$, this implies that in total exactly $\nqf$ districts must win.
   This can only succeed if 
   each component~$C$ is a cycle of length divisible by four
   (also see \cref{fig:RestrTwoFactor} for an illustration).
\end{proof}

\noindent Now, we are ready to show that \probWinner is \NP-complete even for constant
values of $\sold$ and $\Deltas$, except if $\sold=\Deltas=1$, where it is
solvable in polynomial time.

\begin{theorem}\label{thm:biaseddissolution-dichotomy}
\probWinner can be solved in $O(n \cdot(m+n\log{n}))$ time if $\sold=\Deltas=1$;
otherwise it is \NP-complete.
\end{theorem}

\begin{proof}
For $\sold=\Deltas=1$, \probWinner reduces to computing a maximum-weight perfect matching
(see \cref{lem:biased 1-1-dissolution}).
This can be done in $O(n \cdot(m+n\log{n}))$ time~\cite{Gab90}.

It is easy to see that \probWinner is in \NP.
Now, we show the \NP-hardness for $\sold=\Deltas\ge 2$.
For $\sold=\Deltas=2$, observe that \cref{lem:biased 2-2-dissolution} 
implicitly provides a \pt{} reduction from the graph problem
\probRestTwoFactor to \probWinner with $L \subseteq \{3,\dots,|V|\}$.

Two-factors of graphs are computable in polynomial time~\cite{EJ70}.
  However, \probRestTwoFactor is \NP-hard if $(\{3,4,\dots,|V|\} \setminus L) \nsubseteq \{3,4\}$~\cite{HKKK88}.
By \cref{lem:biased 2-2-dissolution}, $(G=(V,E),L)$ with $|V|=4q$ and 
$L=\{4,8,\dots,4q\}$ is a yes-instance of \probRestTwoFactor if and only if
$(G,2,2,q,\w)$ with $\w(v)=1$ for all $v \in V$ is a yes-instance of \probWinner.
  Since $(\{3,4,\dots,|V|\} \setminus \{4,8,\dots,4q\}) \nsubseteq \{3,4\}$ for all $q>1$,
  it follows that \probWinner is \NP-complete when $\sold=\Deltas=2$.
  
  For $\sold=\Deltas\ge3$, we show \NP-hardness by a \pt{} reduction
  from the \NP-complete \probExCover for $t\ge3$ (see the corresponding
  definition in \cref{sec:dissolution-dichotomy}).
  Given an \probExCover instance $(X,\mathcal{C})$ with $|X|=t \cdot q$ elements
  and $r:=|\mathcal{C}|$, we construct a \probWinner instance $(G=(V,E),t,t,\kwin,\w)$. 

  \newcommand{\elementgadget}{$t$-elements gadget}

  To construct the graph~$G$, we use the so-called \emph{\elementgadget{}}.
  A \elementgadget{} consists of a $t$-star where each leaf has an additional degree-one neighbor.
  We call the degree-$t$ vertex \emph{center district}, the original star leaves
  \emph{inner districts}, and the additional degree-one vertices \emph{element districts}.
  A $3$-element gadget is illustrated in \cref{fig:3-element gadget}.
  Now, we add to the graph~$G$ the following:
  \begin{itemize}
    \item $q$ \elementgadget{}s;
          we arbitrarily identify each element~$x \in X$ with exactly one of
          the $(q\cdot t)$ \emph{element districts}
          that is denoted as $v_x$ in the following,
    \item for each subset~$Y \in \mathcal{C}$ a \emph{set district}~$v_{Y}$, and
    \item $r-q$~\emph{dummy districts}.
\end{itemize}
Then, we connect each set district~$v_{Y}$ with each element district~$v_x, x \in Y$,
and connect each dummy district with each set district.
We set %
the number~$\kwin$ of winning districts to $(t+1)\cdot q$.

We now describe how many \partyAsupporters each district
contains~(that is, the function~$\w$).
\begin{itemize}
\item The dummy district contains no \partyAsupporters. 
\item Each set district contains exactly one \partyAsupporter.
\item For each \elementgadget, 
  the center district contains no \partyAsupporters,
  each inner district contains exactly two \partyAsupporters, 
  and each element district contains $t$~\partyAsupporters.
\end{itemize}
This concludes the construction which is illustrated for~$t=3$
in \cref{fig:ExCover-3-3-dissol}.

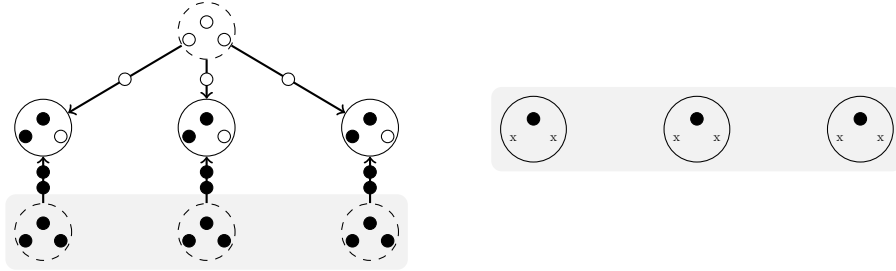
\begin{figure}
   \centering
    \newcommand{\scale}{1}
    \begin{tikzpicture}[scale=\scale]
        \nodeBBBd{c}{};
        \nodeAAB{l2}{below=0.7 of c}
        \nodeAAB{l1}{left=2.0 of l2v1}
        \nodeAAB{l3}{right=2.0 of l2v1}
        \nodeAAAd{u1}{below=0.8 of l1}
        \nodeAAAd{u2}{below=0.8 of l2}
        \nodeAAAd{u3}{below=0.8 of l3}

        \draw [ZArrow] (c) -- (l1) node [midway,name=ec1]{\lB{\scale}};
        \draw [ZArrow] (c) -- (l2) node [midway,name=ec2]{\lB{\scale}};
        \draw [ZArrow] (c) -- (l3) node [midway,name=ec3]{\lB{\scale}};
        \draw [ZArrow] (u1) -- (l1) node [midway,name=e11]{\lAA{\scale}};
        \draw [ZArrow] (u2) -- (l2) node [midway,name=e22]{\lAA{\scale}};
        \draw [ZArrow] (u3) -- (l3) node [midway,name=e33]{\lAA{\scale}};

        \nodeAxx{d1}{right=2 of l3v1}
        \nodeAxx{d2}{right=2.0 of d1v1}
        \nodeAxx{d3}{right=2.0 of d2v1}
     \begin{pgfonlayer}{background}
      \node (d_fit13) [draw=none,transform shape,rounded corners,fill=gray!10,fit=(d1) (d3)] {};
      \node (u_fit13) [draw=none,transform shape,rounded corners,fill=gray!10,fit=(u1) (u3)] {};
     \end{pgfonlayer}
  \end{tikzpicture}
  \caption{Left: A $3$-elements gadget. The only dissolution where A wins all districts 
           requires dissolving the top district and moving exactly one B-supporter
           from the top district to each neighbor.
           Right: Gadget symbol in the construction.}
 \label{fig:3-element gadget}
\end{figure}

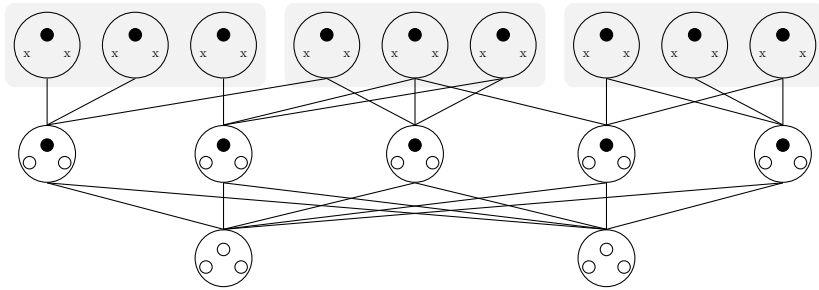
\begin{figure}
   \centering
    \begin{tikzpicture}[scale=1]
      \nodeAxx{e1}{}
      \nodeAxx{e2}{right=1.0 of e1v1}
      \nodeAxx{e3}{right=1.0 of e2v1}
      \begin{pgfonlayer}{background}
       \node (e_fit13) [draw=none,transform shape,rounded corners,fill=gray!10,fit=(e1) (e3)] {};
      \end{pgfonlayer}
      \nodeAxx{e4}{right=1.2 of e3v1}
      \nodeAxx{e5}{right=1.0 of e4v1}
      \nodeAxx{e6}{right=1.0 of e5v1}
      \begin{pgfonlayer}{background}
       \node (e_fit46) [draw=none,transform shape,rounded corners,fill=gray!10,fit=(e4) (e6)] {};
      \end{pgfonlayer}
      \nodeAxx{e7}{right=1.2 of e6v1}
      \nodeAxx{e8}{right=1.0 of e7v1}
      \nodeAxx{e9}{right=1.0 of e8v1}
      \begin{pgfonlayer}{background}
       \node (e_fit79) [draw=none,transform shape,rounded corners,fill=gray!10,fit=(e7) (e9)] {};
      \end{pgfonlayer}

      \nodeABB{S1}{below=0.8 of e1}
      \nodeABB{S2}{below=0.8 of e3}
      \nodeABB{S3}{below=0.8 of e5}
      \nodeABB{S4}{below=0.8 of e7}
      \nodeABB{S5}{below=0.8 of e9}

      \draw [-] (S1.north) -- (e1.south);
      \draw [-] (S1.north) -- (e2.south);
      \draw [-] (S1.north) -- (e4.south);
      \draw [-] (S2.north) -- (e3.south);
      \draw [-] (S2.north) -- (e5.south);
      \draw [-] (S2.north) -- (e6.south);
      \draw [-] (S3.north) -- (e4.south);
      \draw [-] (S3.north) -- (e5.south);
      \draw [-] (S3.north) -- (e6.south);
      \draw [-] (S4.north) -- (e5.south);
      \draw [-] (S4.north) -- (e7.south);
      \draw [-] (S4.north) -- (e9.south);
      \draw [-] (S5.north) -- (e7.south);
      \draw [-] (S5.north) -- (e8.south);
      \draw [-] (S5.north) -- (e9.south);

      \nodeBBB{D1}{below=0.8 of S2}
      \nodeBBB{D2}{below=0.8 of S4}

      \draw [-] (D1.north) -- (S1.south);
      \draw [-] (D1.north) -- (S2.south);
      \draw [-] (D1.north) -- (S3.south);
      \draw [-] (D1.north) -- (S4.south);
      \draw [-] (D1.north) -- (S5.south);
      \draw [-] (D2.north) -- (S1.south);
      \draw [-] (D2.north) -- (S2.south);
      \draw [-] (D2.north) -- (S3.south);
      \draw [-] (D2.north) -- (S4.south);
      \draw [-] (D2.north) -- (S5.south);
 \end{tikzpicture}
  \caption{Illustration of the construction for $t=3$, $r=5$, and $q=3$.}
 \label{fig:ExCover-3-3-dissol}
\end{figure}

  Now, we show that $(X,\mathcal{C})$ is a yes-instance of \probExCover if and only if
  the constructed \probWinner instance $(G,t,t,(t+1)q,\w)$ is a yes-instance.

  \newcommand{\dummy}{\ensuremath{v_{\text{dummy}}}}

  For the ``only if'' part, let $\mathcal{C}' \subseteq \mathcal{C}$ be a subcollection such that
  each element of~$X$ is contained in exactly one subset of~$\mathcal{C}'$.
  A \AbiasedSdissol{(t+1)q}{t}{t} can be constructed as follows.
  Dissolve each center district and move one \partyBsupporter to each of its adjacent inner districts.
  Dissolve each element district and move $(t-1)$ \partyAsupporters to its uniquely determined adjacent inner district.
  For each element district $v_x, x \in X$, move the remaining \partyAsupporter
  to the set district~$v_Y,Y \in \mathcal{C}'$, with $x \in Y$.
  Since $\mathcal{C}'$ partitions $X$, 
  $v_Y$~is uniquely determined.
  The set~$\Vwin$ of winning districts consists of all inner districts and the set districts
  corresponding to the sets in~$\mathcal{C}'$.
  For each dummy district~$\dummy$, uniquely choose one of the set districts~$v_Y, Y\notin\mathcal{C}'$,
  and move all  voters from~$\dummy$ to~$v_Y$.
  This is possible because there are $r-q$ dummy districts and $r-q$ set districts~$v_Y,Y\notin\mathcal{C}'$,
  and each dummy district is adjacent to each set district.
  
  To show that this indeed gives a \AbiasedSdissol{(t+1)q}{t}{t},
  observe that we move all $t$~voters from each dissolved district to
  the adjacent non-dissolved districts.  Each inner district receives
  $\Deltas=t$ voters: $t-1$ \partyAsupporters and one \partyBsupporter.  
  Since each inner district initially contained two
  \partyAsupporters, party~\partyA wins a total of $t\cdot q$ inner districts.
  Each set district $v_Y, Y \in \mathcal{C}'$, receives $t$~\partyAsupporters and
  initially contains one \partyAsupporter.
  Furthermore, $|\mathcal{C}'|=q$, and hence, party~\partyA wins~$q$ set districts in
  total and loses the remaining $r-q$ set districts.
  Thus, we indeed constructed a \AbiasedSdissol{(t+1)q}{t}{t}.

  For the ``if'' part, assume that there is some
  \AbiasedSdissol{(t+1)q}{t}{t} for the constructed instance.
  Since $\sold=\Deltas$ and $G$ has $2t\cdot q+2m$ districts,
  after the dissolution
  a total number of $t\cdot q+r$ districts is dissolved and party~\partyA wins at least $(t+1)q$ districts and loses at most $r-q$ districts.  
  Observe that the only neighbors of the dummy districts are the set districts and hence, by the construction of function~$\w$,
  party~\partyA cannot win any non-dissolved district that receives/contains at least one voter from a dummy district.
  Furthermore, since the set of the $(r-q)$ dummy districts and the set of their neighboring districts build a bipartite induced subgraph,
  there are $(r-q)$ non-dissolved districts which may receive/contain any voters from the dummy districts.
  Thus, party~\partyA loses at least $r-q$ non-dissolved districts.
  Since $\kwin=(t+1)q$, party~\partyA loses \emph{exactly} $r-q$ districts.
  In particular, each of the losing districts contains at least one voter (originally) from a dummy district.
  This implies that party~\partyA has to win each non-dissolved set district, 
  element district, inner district, or center district.
  However, the construction of $\w$ forbids \partyA to win a
  center district or to win an inner district if one
  moves two \partyBsupporters to it.  Thus, we dissolve each center district 
  and move exactly one \partyBsupporter from this center district to each of its adjacent inner districts. 
  As a direct consequence, all element districts are to be dissolved and $t-1$ voters are moved from each element district to its adjacent inner districts such that~\partyA wins all $t\cdot q$ inner districts.  
  There are $t\cdot q$ \partyAsupporters left, one \partyAsupporter from each element district.
  These voters are to be moved to a set of exactly $q$ winning set districts each.
  Since each of these districts needs at least $t$ \partyAsupporters to win and
  has exactly $t$~adjacent element districts, $C':=\{S \in \mathcal{C} \mid v_S \in \Vwin\}$ partitions~$X$.
\end{proof}

\section{Special graph classes}
\label{sec:graph-classes}
First, in \cref{sec:planar}, we consider \probWinner{} on planar
graphs. This problem restriction is interesting especially in the
political districting context since the neighborhood relation between
voting districts on a map is typically planar.  We will see that
\probValid{} and, thus, \probWinner{}, unfortunately remains \NP-hard
for many choices of~$s$ and~$\Delta_s$.

Second, in \cref{sec:clique}, we show that \probWinner{} is \pt{}
solvable on cliques, that is, if voters may be moved unrestrictedly
between dissolved districts and non-dissolved districts.

Finally, in \cref{sec:tw}, we consider \probWinner{} on graphs of
bounded treewidth. This problem restriction is interesting in the
context of distributed systems since computers are often
interconnected using a tree, star, or bus topology.  By presenting a
formulation of \probWinner{} in the monadic second-order logic of
graphs, we show that \probWinner{} is solvable in linear time on
graphs of bounded treewidth when~$s$~and~$\Delta_s$ are
constant. This, however, should be understood as a pure classification
result rather than as an implementable algorithm.

\subsection{Planar graphs}\label{sec:planar}
Computing star partitions is known to be \NP-hard 
even on subcubic grid graphs and split graphs~\cite{BBBCFNW14}.
By \cref{prop:starpacking-relation} in \cref{sec:Relation2Starmatching}
it follows that \probValid is also \NP-hard on planar graph because grid graphs are planar.
However, the \NP-hardness reduction on subcubic grid graphs requires
stars with two leaves such that the \NP-hardness does only transfer
to computing \Sdissol{1}{2}{}s.
Here, we show that \NP-hardness for \probValid{} holds for any constants~$\sold$
and~$\Deltas$ such that $\Delta_s$~divides~$s$ or $s$~divides~$\Delta_s$.

By giving a \pt{} %
reduction from the following
\NP-complete problem, it is easy to derive \NP-hardness results for
\probValid{}.

\decprob{\probPlanarPacking{$H$}}
{A planar undirected graph~$G=(V, E)$.}
{Does $G$ contain an
  \emph{$H$-factor}~$V_1,V_2,\ldots,V_{\lceil|V|/|V(H)|\rceil}$ that
  partitions the vertex set~$V$ such that $G[V_i]$ is isomorphic to
  $H$ for all $i$?}

\noindent \probPlanarPacking{$H$} is \NP-complete for any connected
outerplanar graph~$H$ with three or more vertices~\cite{BJLSS90}.  In
particular, \probPlanarPacking{$H$} is \NP-complete for any $H$ being a
star of size at least three.
This makes it easy to prove the following theorem:

\begin{theorem}
  \label{thm:planar-d|delta}
  \probValid on planar graphs is \NP-complete for all~$s\neq\Delta_s$
  such that $\Delta_s$~divides~$s$ or $s$~divides~$\Delta_s$. It is
  \pt{} solvable for~$s=\Delta_s$.
\end{theorem}

\begin{proof}
  We have already shown in \cref{thm:dissolution-dichotomy} how to
  solve \probValid{} in polynomial time for~$s=\Delta_s$.
  Hence, now  assume that $\Deltas\neq s$ and~$s$ divides~$\sold$.
  Let $x:=\sold / \Deltas\ge 2$.
  Due to \cref{prop:starpacking-relation} and the fact
  that \probPlanarPacking{$K_{1,x}$} is \NP-complete~\cite{BJLSS90} we can
  conclude that \probValid is \NP-complete even on planar graphs.
\end{proof}

\noindent
It seems to be challenging to transfer the dichotomy
result for \probValid{} on general graphs
(\cref{thm:dissolution-dichotomy}) to the case of planar
graphs. The main problem is that the proof of
\cref{thm:dissolution-dichotomy} exploits \probExCover{} to be
\NP-hard for all~$t\geq 3$. The reduction from \probExCover{} to
\probValid{} produces a graph that contains the
incidence graph of the \probExCover{} instance as a subgraph. %
To obtain a reduction to \probValid on planar graphs, 
it is necessary to have planar incidence graphs of \probExCover.
It is,
however, unknown whether this problem variant, called \textsc{Planar}
\probExCover{}, is \NP-hard for~$t\geq 4$. One might be misled to
think that \probExCover{} is \NP-hard for~$t\geq 4$ since it already is
\NP-hard for~$t=3$. However, the closely related problem \textsc{Planar
  3-Sat}, that is, \textsc{3-Sat} with planar clause-literal incidence
graphs, is \NP-complete, whereas \textsc{Planar 4-Sat} is 
\pt{} solvable: one can show that the
clause-literal incidence graph of a \textsc{Planar 4-Sat} instance
allows for a matching such that each clause is matched to some
literal. These literals can then be simply set to true in order to
satisfy all clauses.
We consider the question whether \textsc{Planar Exact Cover by 4-Sets}
is \NP-hard of independent interest.

\subsection{Cliques}\label{sec:clique}
If the neighborhood graph is a clique, that is, the districts are
fully connected such that voters can \move from any dissolved district
to any non-dissolved district, then the existence of an \dissol
depends only on the number~$|V|$ of districts, the district size~$s$,
and the size increase~$\Deltas$.  Clearly, a \probValid instance is a
yes-instance if and only if~$d := |V|\cdot\Deltas/(s+\Deltas)$ is an
integer.  We now show that \probWinner is not as easy but still
solvable in polynomial time if the neighborhood graph is a clique.
The basic idea is to dissolve districts with a large number
of~\partyAsupporters while minimizing the number of losing districts
by letting the districts with the smallest number of \partyAsupporters
lose.

\begin{theorem}\label{thm:cliquesP}
  \probWinner on cliques is solvable in~$O(|V|^2)$ time.
\end{theorem}

\begin{proof}
  In fact, we show how to solve the optimization version of
  \probWinner, where we maximize the number~$\kwin$ of winning
  districts.  Intuitively, it appears to be a reasonable approach to
  dissolve districts pursuing the following two objectives.  Our
  \emph{first objective} is that any losing district should contain as
  few \partyAsupporters as possible.  Our \emph{second objective} is
  that any winning district should contain only as
  many \partyAsupporters as necessary.  Dissolving districts this way
  minimizes the number of ``wasted'' \partyAsupporters.

  We now show that this greedy strategy is indeed optimal.  To this
  end, let $G=(V,\binom{V}{2})$ be a clique, let $\w$ be
  an \partyAsupporter distribution over $V$, and let $\sold$ and
  $\Deltas$~be the district size and the district size increase.  With
  $G$ being a complete graph, we are free to move voters from any
  dissolved district to any non-dissolved district.  Let~$\mu :=
  \lfloor (s+\Deltas)/2\rfloor+1$ be the minimum number
  of \partyAsupporters required to win a district.  Thus, a district
  with less than~$(\mu - \Deltas)$
  \partyAsupporters can never win. 
  Define~$\mathcal{L}:=\{v\in V \mid \w(v)< \mu - \Deltas\}$
  to be the set of \emph{non-winnable} districts.

  Our strategy can be sketched as follows
  (see also \cref{fig:winnable-non-winnable} for an illustration).
  Assume that $\kdis$~districts have to be dissolved and $\ell$~districts have to lose,
  and let $\mu$~denote the number of \partyAsupporters needed to win a district.
  Sort the districts according to the number of \partyAsupporters.
  Mark the $\ell$~districts with the fewest number of \partyAsupporters as losing.
  Dissolve all non-marked non-winnable districts.
  If necessary, then also dissolve winnable districts beginning with those
  with the most \partyAsupporters until~$\kdis$ districts have been dissolved.
  Finally, check whether this gives a solution.

  \begin{figure}
   \centering
    \begin{tikzpicture}[scale=0.45,decoration={brace,amplitude=2pt}]
      \begin{scope}
      \draw[->] (0,0) -- (0,7);
      \draw[->] (0,0) -- (10,0);

      \draw[-,dotted] (-0.2,5) -- (10,5);
      \node[draw=none,anchor=east] at (-0.2,5) {$\mu$};
      \draw[-,dotted] (-0.2,2.5) -- (10,2.5);
      \node[draw=none,anchor=east] at (-0.2,2.5) {$\mu - \Deltas$};

      \draw [fill=gray!50] (0.8,0) rectangle (1.2,0.05);
      \draw [fill=gray!50] (1.8,0) rectangle (2.2,0.5);
      \draw [fill=gray!50,pattern=north west lines] (2.8,0) rectangle (3.2,1);
      \draw [fill=gray!50,pattern=north east lines] (3.8,0) rectangle (4.2,1.5);
      \draw [fill=gray!50] (4.8,0) rectangle (5.2,2.5);
      \draw [fill=gray!50] (5.8,0) rectangle (6.2,2.5);
      \draw [fill=gray!50] (6.8,0) rectangle (7.2,3);
      \draw [fill=gray!50] (7.8,0) rectangle (8.2,3.5);
      \draw [fill=gray!50,pattern=dots] (8.8,0) rectangle (9.2,6);

      \draw [decorate,decoration={brace,amplitude=5pt}] (4.2,0) -- (0.8,0)
            node [midway,anchor=north,inner sep=5pt, outer sep=5pt]{non-winnable};
      \draw [decorate,decoration={brace,amplitude=5pt}] (9.2,0) -- (4.8,0)
            node [midway,anchor=north,inner sep=5pt, outer sep=5pt]{winnable};
      \end{scope}
      \begin{scope}[xshift=12.5cm]
      \draw[->] (0,0) -- (0,7);
      \draw[->] (0,0) -- (10,0);

      \draw[-,dotted] (-0.2,5) -- (10,5);
      \draw[-,dotted] (-0.2,2.5) -- (10,2.5);

      \draw [fill=gray!50] (0.8,0) rectangle (1.2,0.05);
      \draw [fill=gray!50] (1.8,0) rectangle (2.2,0.5);
      \draw [fill=gray!50] (4.8,0) rectangle (5.2,2.5);
      \draw [fill=gray!50,pattern=north east lines] (4.8,2.5) rectangle (5.2,4);
      \draw [fill=gray!50,pattern=north west lines] (4.8,4) rectangle (5.2,5);
      \draw [fill=gray!50] (5.8,0) rectangle (6.2,2.5);
      \draw [fill=gray!50,pattern=dots] (5.8,2.5) rectangle (6.2,5);
      \draw [fill=gray!50] (6.8,0) rectangle (7.2,3);
      \draw [fill=gray!50,pattern=dots] (6.8,3) rectangle (7.2,5);
      \draw [fill=gray!50] (7.8,0) rectangle (8.2,3.5);
      \draw [fill=gray!50,pattern=dots] (7.8,3.5) rectangle (8.2,5);

      \draw [decorate,decoration={brace,amplitude=5pt}] (2.2,0) -- (0.8,0)
            node [midway,anchor=north,inner sep=5pt, outer sep=5pt]{losing};
      \draw [decorate,decoration={brace,amplitude=5pt}] (8.2,0) -- (4.8,0)
            node [midway,anchor=north,inner sep=5pt, outer sep=5pt]{winning};
    \end{scope}
    \end{tikzpicture}
    \caption{Assume that we want to find a 4-biased \dissolution for
      the instance illustrated on the left hand side, where each
      district is represented by a bar of height proportional its
      number of \partyAsupporters.  Following our first objective,
      we dissolve the two non-winnable districts with the
      most \partyAsupporters and, following our second objective,
      we dissolve the winnable district with the most \partyAsupporters.
      Non-dissolved districts are represented by bars filled with
      solid gray. Each dissolved district is represented by a bar
      filled with an individual pattern.  The diagram on the right
      illustrates the solution, where the \partyAsupporters of the two
      dissolved non-winnable districts moved to the first winnable
      district and the \partyAsupporters of the dissolved winnable
      district moved to the three remaining winnable districts.  }
   \label{fig:winnable-non-winnable}
  \end{figure}
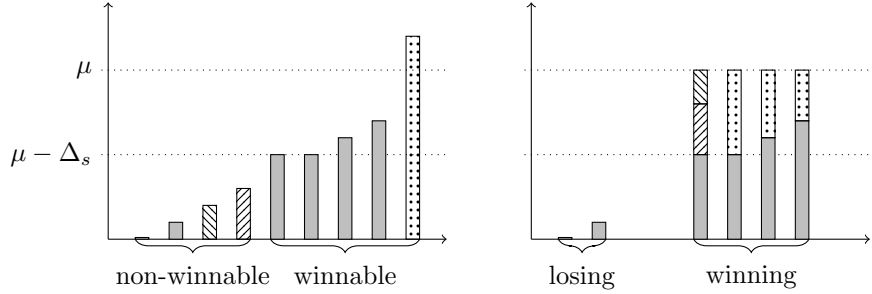

  Our first claim corresponds to the first objective above, that is,
  the losing districts should contain a minimal number of~\partyA-supporters.
  \begin{claim}
    \label{clm1}
    Let $v,w\in V$ be two districts with~$\w(v)\le\w(w)$.
    If there exists an~$\kwin$-biased dissolution where~$v$ is winning
    and~$w$ is losing, then there also exists an~$\kwin$-biased
    dissolution where~$v$ is losing and~$w$ is winning.
  \end{claim}
  \noindent To verify \cref{clm1}, let $(\Vdis,\z,\za,\Vwin)$ be an
  $\kwin$-biased dissolution.
  Let~$v\in\Vwin$ and~$w\in (V\setminus\Vdis) \setminus \Vwin$ be two districts such that~$\w(v)\le\w(w)$.
  Now, simply exchange~$v$ and~$w$, that is,
  set~$\Vwin':=(\Vwin\setminus\{v\})\cup\{w\}$
  and define for all~$(x,y)\in \A{\Vdis, G}$:
  \[\z'(x,y):=\begin{cases}
              \z(x,w)\;& \text{if } y=v\\
              \z(x,v)\;& \text{if } y=w\\
              \z(x,y)\;&\text{else,}
              \end{cases}\quad
  \za'(x,y):=\begin{cases}
             \za(x,w)\;& \text{if } y=v\\
             \za(x,v)\;& \text{if } y=w\\
             \za(x,y)\;&\text{else.}
             \end{cases}\]
  Since~$\w(v) \le \w(w)$, it is clear that~$(\Vdis,\z',\za',\Vwin')$
  is also a well-defined~$\kwin$-biased dissolution.

  The next claim basically corresponds to the second objective above, in the
  sense that districts with a large number of \partyA-supporters
  (possibly too large, that is, more than the required~$\mu$) should
  be dissolved in order to move the voters more efficiently.
  \begin{claim}
    \label{clm2}
    Let $v,w\in V$ be two districts with~$\w(v)\le\w(w)$.
    Assume that there exists an \biaseddissol with maximum~$\kwin$.
    If~$v$ is dissolved, then the following holds:
    \begin{enumerate}[(i)]
    \item\label{clm21} If~$w$ is losing, then there also exists an~$\kwin$-biased
    dissolution where~$w$ is dissolved and~$v$ is losing.
    \item\label{clm22} If~$w$ is winning and~$v$ is winnable, that is,
    $v\not\in\mathcal{L}$, then there exists an~$\kwin$-biased
    dissolution where~$w$ is dissolved and~$v$ is winning.
    \end{enumerate}
  \end{claim}

  \noindent \cref{clm2} also holds by an exchange argument similar to
  the one above:
  Let $(\Vdis,\z,\za,\Vwin)$ be an~$\kwin$-biased
  dissolution and let~$v\in\Vdis$, $w\in V\setminus\Vdis$ be two districts such
  that~$\w(v)\le\w(w)$.
  Again, we exchange~$v$ and~$w$ by setting~$\Vdis' := \Vdis \setminus
  \{v\}\cup\{w\}$.
  Since~$\sum_{x\in\Vdis'}\w(x) \ge \sum_{x\in\Vdis}\w(x)$ and since we are
  free to \move voters arbitrarily between districts, it is clear
  that it is always possible to find an~$\kwin$-biased dissolution such that
  $\Vdis'$ is the set of dissolved districts. 
  In particular, if~$v$ is a winnable district, then it is always
  possible to make~$v$ a winning district.

  Using \cref{clm1,clm2} above, we now show how to compute an optimal
  biased dissolution.
  In order to find a biased dissolution with the maximum number
  of winning districts, we search for a dissolution that loses
  a minimum number of remaining districts.
  Thus, for each~$\ell\in\{0,\ldots,\krem\}$, we check whether it is possible to
  dissolve $\kdis$~districts such that at most~$\ell$ of the remaining~$\krem$
  districts lose.
  To this end, assume that the districts~$v_1,\ldots,v_n$
  are ordered by increasing number of \partyAsupporters, that
  is, $\w(v_1)\leq\w(v_2)\leq \ldots \leq \w(v_n)$
  and let~$V_{\ell}:=\{v_1,\ldots,v_{\ell}\}$.
  Now, if there exists an~$(\krem-\ell)$-biased dissolution, then there also
  exists an~$(\krem-\ell)$-biased dissolution where the losing
  districts are exactly~$V_\ell$. This follows by repeated application
  of the exchange arguments of \cref{clm1} and
  \cref{clm2}(\ref{clm21}).
  Hence, given~$\ell$, we have to check whether there is a
  set~$D\subseteq V\setminus V_\ell$ of~$\kdis$ districts
  that can be dissolved in such a way that all non-dissolved districts
  in~$V\setminus (V_\ell \cup D)$ win and the districts in~$V_\ell$ lose. 

  \newcommand{\Vlastd}{\ensuremath{V^{d'}}}
  
  First, note that in order to achieve this, all districts in~$\mathcal{L}\setminus
  V_\ell$ have to be dissolved because they cannot win in any way.
  Clearly, if~$|\mathcal{L}\setminus V_\ell|>\kdis$, then it is simply not possible
  to lose only~$\ell$~districts and we can immediately
  go to the next iteration with $\ell:=\ell+1$.
  Therefore, we assume that~$|\mathcal{L}\setminus V_\ell| \le \kdis$ and
  let~$\kdis':=\kdis-|\mathcal{L}\setminus V_{\ell}|$ be the number of additional
  districts to dissolve in~$V\setminus(\mathcal{L}\cup V_{\ell})$.
  By \cref{clm2}(\ref{clm22}), it follows that we can assume that
  the $d'$~districts with the maximum number of~\partyA-supporters are
  dissolved, that is, $\Vlastd := \{v_{n-d'+1},\ldots,v_n\}$.
  Thus, we set~$\Vdis:=\mathcal{L}\setminus V_\ell \cup
  \Vlastd$ and check whether there are enough \partyAsupporters in~$\Vdis$
  to let all~$\krem-\ell$ remaining districts in~$V\setminus (V_\ell\cup \Vdis)$~win.
  
  Sorting the districts by the number of \partyAsupporters (in a
  preprocessing step) requires $O(n \log n)$ comparisons.  Then, for
  up to $n$~values of $\ell$, to check whether the remaining districts
  in $V\setminus (V_\ell\cup \Vdis)$ can win requires $O(n)$
  arithmetic operations each.  Thus, assuming constant-time
  arithmetic, we end up with a total running time in $O(n^2)$.
\end{proof}

\subsection{Graphs of bounded treewidth}
\label{sec:tw}

\citet[Theorem 2.3]{Yus07} showed that \textsc{$H$-Factor} is solvable
in linear time on graphs of bounded treewidth when the size of~$H$ is
constant.  This includes the case of finding $x$-star partitions, that
is, \Sdissol{x}{1}{}s respectively \Sdissol{1}{x}{}s when $x$~is
constant. We can show that the more general problem \probWinner{} is
solvable in linear time on graphs of bounded treewidth when~$s$
and~$\Delta_s$ are constants.  In terms of parameterized complexity
analysis~\citep{DF13,FG06,Nie06}, this shows that \probWinner{} is
fixed-parameter tractable with respect to the combined
parameter~$(t,s,\Delta_s)$, where $t$~is the treewidth of the
neighborhood graph.  Note that these results are basically for
classification only, since the corresponding algorithms come along
with enormous constants hidden in the $O$-notation.

\begin{theorem}\label{fplin}
  \probWinner{} is solvable in linear time on graphs of constant
  treewidth when $s$ and~$\Delta_s$ are constants.
\end{theorem}

\noindent
To prove \cref{fplin}, we exploit a general result that a
maximum-cardinality set satisfying a constant-size formula in monadic
second-order logic for graphs can be computed in linear time on graphs
of constant bounded treewidth~\citep{ArnborgLS91}.
The set whose size we want to maximize is the set~$\Vwin$ of winning districts.
For the remainder of this subsection, we consider multigraphs,
that is, our graphs may contain multiple edges between any two vertices.

\begin{definition}[Monadic second-order logic for graphs] A
  formula~$\phi$ of the monadic second-order logic for graphs may
  consist of the logic operators~$\vee,\wedge,\neg$, vertex variables,
  edge variables, set variables, quantifiers~$\exists$ and~$\forall$
  over vertices, edges, and sets, and the predicates
  
  \begin{enumerate}[i)]
  \item $x\in X$ for a vertex or edge variable~$x$ and a set~$X$,
  \item $\inc(e,v)$, being true if $e$~is an edge incident to the vertex~$v$,
  \item $\adj(v,w)$, being true if $v$~and~$w$ are adjacent vertices,
  \item equality of vertex variables, edge variables, and set
    variables.
  \end{enumerate}
\end{definition}

\noindent
We point out that a constant-size formula in monadic second-order
logic for a problem does not only prove the mere existence of a
linear-time algorithm on graphs of bounded treewidth; the formula
itself can be converted into a linear-time
algorithm~\citep[Chapter~6]{CE12}.

\begin{proof}[Proof of \cref{fplin}]
  We model \probWinner{} as a formula in monadic second-order logic.
  Since monadic second-order logic does not allow
  to count the number of voters moved from one district to another
  or to count how many \partyAsupporters a district contains,
  we first model \probWinner{} as a problem on an auxiliary graph. For
  constant $s$~and~$\Delta_s$, the transformation of a \probWinner{}
  instance to this auxiliary graph can be done in linear time and
  works as follows (see \cref{fig:trafo}):

\begin{enumerate}
\item For each input district of \probWinner{}, introduce a vertex and
  attach to it as many degree-one vertices as the district has
  \partyAsupporters.
\item Between two neighboring districts, add~$s+1$ (multiple) edges
  between their representing vertices. The $s+1$ (multiple) edges
  represent potential moves of voters from one district to another.
\item Finally, connect each pair of vertices representing a pair of
  neighboring districts by $s$~parallel subdivided edges. These
  represent potential moves of \partyAsupporters.
\end{enumerate}

\noindent Note that, by adding~$s+1$ (multiple) edges between any two vertices
representing neighboring districts, we ensure that, in the graph resulting
from the above construction, a vertex has degree one if and only if it
represents an \partyAsupporter.
The vertex representing an \partyAsupporter belongs to the district
represented by its neighbor.
Moreover, a vertex has degree two if and only if it
represents a possible movement of an \partyAsupporter of one
district to another.

\begin{figure}
  \centering
    \begin{tikzpicture}[scale=0.9]
      \begin{scope}
        \nodeAAB{v1}{}
        \nodeABB{v2}{below left=of v1}
        \nodeABB{v3}{below right=of v1}

        \draw [ZEdge] (v1) -- (v2);
        \draw [ZEdge] (v1) -- (v3);
      \end{scope}
    \end{tikzpicture}\hspace{2cm}
    \begin{tikzpicture}[thick]
    \tikzstyle{vertex}=[circle,draw,fill=black,minimum size=5pt,inner sep=0pt]
    \tikzstyle{svertex}=[circle,draw,fill=black,minimum size=3pt,inner sep=0pt]

    \node[vertex] at (2,2) (v1) {};
    \node[vertex] at (0,0) (v2) {};
    \node[vertex] at (4,0) (v3) {};
    \node[vertex] at (2.5,2.5) (pen1) {};
    \node[vertex] at (1.5,2.5) (pen2) {};
    \node[vertex] at (-0.5,-0.5) (pen3) {};
    \node[vertex] at (4.5,-0.5) (pen4){};

    \draw  (v1) -- (v2);
    \draw  (v1) to[out=180,in=90] (v2);
    \draw  (v1) to[out=195,in=75] (v2);
    \draw  (v1) to[out=210,in=60] (v2);

    \draw  (v1) to[out=-90,in=0] node[svertex, midway] {} (v2);
    \draw  (v1) to[out=-105,in=15] node[svertex, midway]{} (v2);
    \draw  (v1) to[out=-120,in=30] node[svertex, midway]{} (v2);

    \draw  (v1) -- (v3);
    \draw  (v1) to[out=0,in=90] (v3);
    \draw  (v1) to[out=-15,in=105] (v3);
    \draw  (v1) to[out=-30,in=120] (v3);

    \draw  (v1) to[out=-60,in=150] node[svertex, midway]{}  (v3);
    \draw  (v1) to[out=-75,in=165]  node[svertex, midway]{} (v3);
    \draw  (v1) to[out=-90,in=180]  node[svertex, midway]{} (v3);

    \draw (v2)--(pen3);
    \draw (v3)--(pen4);
    \draw (v1)--(pen1);
    \draw (v1)--(pen2);
    \end{tikzpicture}

    \caption{Illustration of transforming a \probWinner{} instance
      (left) into an instance of the auxiliary multigraph problem (right).}
  \label{fig:trafo}
\end{figure}
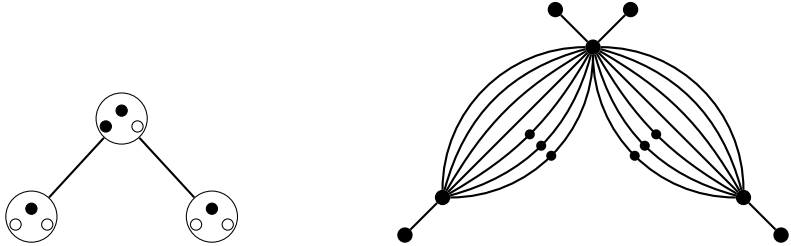

A dissolution now does not contain a function~$z$ moving voters from
one district to another (see \cref{def:dissolution}), but a set~$Z$
of selected edges representing such movements. Similarly,
the \partyAsupporter movement is no longer modeled as a
function~$z_\alpha$ (see \cref{def:biased}), but as a set of
vertices~$Z_\alpha$ representing such movements. Hence, we search for
a maximum vertex-set~$\Vwin$ that satisfies the following formula in monadic
second-order logic of graphs:
\begin{align*}
  \max \Vwin&\text{ s.\,t.\ }\begin{aligned}[t] \exists D\exists
    Z\exists
    Z_\alpha[\text{movements}\wedge\text{\partyA-movements}\wedge\text{districts}\\
    {}\wedge \text{prop-}a\wedge \text{prop-}b\wedge \text{prop-}c \wedge \text{prop-}d\wedge \text{prop-}e],
  \end{aligned}
  \intertext{where $\text{prop-}a, \text{prop-}b, \text{prop-}c, \text{prop-}d$,
    and $\text{prop-}e$ will be predicates ensuring that the
    Properties a)--e) of Definitions~\ref{def:dissolution}
    and~\ref{def:biased} of (biased) \dissolution are satisfied, $D$~will be the set of
    dissolved districts, $Z$~the set of voter movements, and~$Z_\alpha$
    the set of \partyAsupporter movements. To ensure this, we
    define}
  \text{districts}&:=\forall v[(v\in \Vdis \vee v \in \Vwin)\implies
  \text{degree-greater-two}(v)]
  \intertext{so that it is true if and only if each element in~$\Vdis \cup \Vwin$
    is a vertex with degree more than two, that is, it represents a district, where}
  \text{degree-greater-two}(v)&:=\exists v_1\exists v_2\exists
  v_3[\begin{aligned}[t]&v_1\neq v_2\wedge v_1\neq v_3\wedge v_2\neq v_3\\
    & \wedge \adj(v_1,v) \wedge \adj(v_2,v) \wedge \adj(v_3,v)]
  \end{aligned}
  \intertext{is true if and only if $v$~has at least three
    neighbors. Moreover, we define}
  \text{movements}&:=\forall e[\begin{aligned}[t] &e\in
    Z\implies \exists v_1\exists v_2[
    \inc(e,v_1)\wedge\inc(e,v_2)\\
    &\wedge\text{degree-greater-two}(v_1)
    \wedge v_1\in D\\
    &\wedge\text{degree-greater-two}(v_2) \wedge v_2\notin D]]
  \end{aligned}
  \intertext{so that it is true if and only if each element in the
    set~$Z$ is an edge representing a movement and}
  \text{\partyA{}-movements}&:=\forall a[\begin{aligned}[t] &a\in
    Z_{\alpha}\implies \exists v_1\exists v_2[
    \adj(a,v_1)\wedge\adj(a,v_2)\\
    &\wedge v_1\in \Vdis \wedge v_2\notin \Vdis\wedge\neg \text{degree-greater-two}(a)]]
  \end{aligned}
  \intertext{so that it is true if and only if each element in the
    set~$Z_{\alpha}$ is a vertex representing a movement of an \partyAsupporter. It remains
    to give the definitions of the predicates~$\text{prop-}a, \text{prop-}b, \text{prop-}c, \text{prop-}d,$
    and~$\text{prop-}e$. We define}
  \text{prop-}a:=\forall v[%
    v\in \Vdis{}&\implies \exists Z'[
    \card_s(Z')\wedge(\forall e[e\in Z'\iff
    \text{move-from}(e,v)])]]
  \intertext{so that it is true if and only if for each dissolved district~$v$
    there is a set of $\sold$~edges representing movements out of~$v$,
    where}
  \text{move-from}(e,v)&:=v\in \Vdis\wedge e\in Z\wedge\inc(e,v)
  \intertext{is true if and only if~$e$ is an edge representing a movement out of~$v$ and}
  \text{card}_{i}(X)&:=\exists x_1\exists x_2\dots\exists
  x_i\Biggl[
  \begin{aligned}[t]
    \Bigl(\bigwedge_{j=1}^ix_i\in X\Bigr)\wedge
    \Bigl(\bigwedge_{j=1}^i \bigwedge_{k=j+1}^i (x_j\neq x_k)\Bigr)\\
    {}\wedge{}\forall x\Bigl[x\in X\implies \bigvee_{j=1}^ix_j=x\Bigr]
    \Biggr]\\
  \end{aligned}
  \intertext{for $1\leq i\leq s$~is a constant-size formula that is
    true if and only if the set~$X$ has cardinality~$i$. Next, we define}
  \text{prop-}b&:= \forall v[\begin{aligned}[t]
    &(\text{degree-greater-two}(v)\wedge v\notin \Vdis)\implies\\&\exists Z'[
      \card_{\Delta_s}(Z')
      {}\wedge{}(\forall e[e\in Z'\iff
    \text{move-to}(e,v)])]]
  \end{aligned}
  \intertext{so that it is true if and only if there is a set~$Z'$ of
    $\Delta_s$~edges representing movements to each non-dissolved district~$v$, where}
  \text{move-to}(e,v)&:= v\notin D\wedge e\in Z\wedge\inc(e,v)
  \intertext{is true if and only if~$e$ is an edge representing a movement to~$v$.
   Next, we define}
  \text{prop-}c:=\forall v\forall u[
    v\in D {}&\wedge u\notin D \begin{aligned}[t] \wedge \adj&(v,u)\implies
    \exists Z'\exists Z'_\alpha[\text{smaller-equal}(Z'_\alpha,Z')\\
    &\wedge (\forall e[e\in Z'\iff
    \begin{aligned}[t]
      &\text{move-from}(e,v)\\
      &{}\wedge{}\text{move-to}(e,u)])
    \end{aligned}\\
    &\wedge (\forall a[a\in Z'_\alpha\iff
    \begin{aligned}[t]
      &\text{A-move-from}(a,v)\\
      &{}\wedge{}\text{A-move-to}(a,u)])]]
    \end{aligned}
  \end{aligned}
  \intertext{so that it is true if and only if the number of vertices
    representing \partyAsupporters movements from~$v$ to~$u$ is at most
    the number of edges representing movements from~$v$ to~$u$, where}
  \text{smaller-equal}(X,Y)&:=\bigvee_{i=1}^s\bigvee_{j=i}^s
  (\card_i(X)\wedge\card_j(Y))
  \intertext{is a constant-size formula that is true if and only if
    $|X|\leq |Y|$ and}
  \text{A-move-from}(a,v)&:=
  \begin{aligned}[t]
    v\in D\wedge a\in Z_\alpha\wedge\adj(v,a),
  \end{aligned}
  \\
  \text{A-move-to}(a,u)&:=
  \begin{aligned}[t]
    u\notin D\wedge a\in Z_\alpha\wedge\adj(u,a)
  \end{aligned}
  \intertext{are true if and only if $a$~is a vertex representing
    an \partyAsupporter movement from~$v$ or to~$u$, respectively. Next, we define}
  \text{prop-}d:=
    \forall v[v\in D&\implies\begin{aligned}[t]
    \exists Z_\alpha'\exists A[&\text{equal-card}(Z_\alpha',A)\\
    &{}\wedge{}\forall a[a\in A\iff\text{A-supporter-of}(a,v)]\\
    &{}\wedge{} \forall a[a\in Z'_\alpha\iff
    \text{A-move-from}(a,v)]]]
  \end{aligned}
  \intertext{so that it is true if and only if the number
    of \partyAsupporter movements out of a district~$v$ equals the
    number of its \partyAsupporters, where}
  \text{equal-card}(X,Y)&:=
  \bigvee_{i=1}^s(\card_i(X)\wedge\card_i(Y))
  \intertext{is a constant-size formula that is true if and only
    if~$|X|=|Y|$ and}
  \text{A-supporter-of}(a,v)&:=
  \begin{aligned}[t]
    \adj(a,v)\wedge\forall u[\adj(a,u)\implies u=v]
  \end{aligned}
  \intertext{is true if and only if $a$~is vertex representing an \partyAsupporter in
    district~$v$. Finally, we define}
  \text{prop-}e:=
    \forall v[v\in \Vwin&\implies \begin{aligned}[t] \exists
    A[&\text{card}_{> (s+\Delta_s)/2}(A)\\
    &{}\wedge\forall a[a\in A\iff
    \begin{aligned}[t]
    &\text{\partyAsupporter-of}(a,v)\\
    &{}\vee{}\text{\partyA{}-move-to}(a,v)]]]
  \end{aligned}
  \end{aligned}
  \intertext{so that for each district $v \in \Vwin$
    there are more than~$(s+\Delta_s)/2$ vertices which
    either represent \partyAsupporters of district~$v$
    or represent \partyAsupporter movements to district~$v$,
    where}
  \text{card}_{> i}(X)&:=\smashoperator{\bigvee_{j=\lfloor i\rfloor+1}^{\sold+\Deltas}}\card_j(X)
\end{align*}
  is a constant-size formula that is true if and only
  if~$i<|X|\le\sold+\Deltas$ with $i<\sold+\Deltas$. 
\end{proof}

\noindent Without providing any details, we claim that one can also prove
\cref{fplin} by using an explicit dynamic programming algorithm that
works on a so-called tree decomposition of a graph. The algorithm runs
in~$(\Delta_s+s)^{O(t^2)}\cdot n^{O(1)}$~time, but it is very technical and its
correctness proof is very tedious, while practical applicability still
seems out of reach.

\section{Conclusion}
We initiated a graph-theoretic approach to concrete
redistribution problems with potential applications in such diverse areas as 
political districting, green computing, and economization of work processes. 
Obviously, the two basic problems \probValid{} and \probWinner{} concern
highly simplified situations and will not be able to model all interesting
aspects of redistribution scenarios.
For instance, our constraint that before and after the dissolution all vertex 
loads are perfectly balanced may be too restrictive for many applications. 
All in all, we consider our simple (and yet fairly realistic) models as a 
first step into a promising direction for future research.
In particular, this may yield a stronger linking of graph-theoretic concepts
with districting scenarios and other application scenarios.

\looseness=-1 We end with a few specific challenges for future research.  We left
open whether the \classP~vs.~\NP dichotomy for general graphs fully
carries over to the planar case: it might be possible that planar
graphs allow for some further tractable cases with respect to the
relation between old and new district sizes.  To this end, it might
help to answer the question whether \textsc{Planar Exact Cover by
  4-Sets} is \NP-hard.  Since \textsc{Planar Exact Cover by 4-Sets} is
a very natural and simple problem on planar graphs, we believe that
this question is of independent interest.  Moreover, with
redistricting applications in mind it might be of interest to study
special cases of planar graphs (such as grid-like structures) in quest
of finding \pt{} solvable special cases of network-based vertex
dissolution problems.  Having identified several \NP-complete special
cases of \probValid{} and \probWinner{}, it is a natural endeavor to
investigate their \pt{} approximability and their parameterized
complexity; in the latter case one also needs to identify fruitful
parameterizations.  Motivated by our results, parameters measuring the
distance to acyclic graphs (cf.\ \cref{fplin}) or to complete graphs
(cf.\ \cref{thm:cliquesP}) seem promising in the spirit of distance
from triviality parameterizations~\citep{GHN04, Nie10}.  Furthermore,
also the maximum degree of a vertex should not be too large in many
applications.  On the one hand, since already the partition of a graph
into paths of length three, which is a special case of our \probValid
problem, is \NP-hard on graphs with maximum degree at most
three~\cite{MZ05,MT07}, the parameter ``maximum degree'' is not
interesting as single parameter.  On the other hand, the maximum
degree might be worth to be considered in combination with other
parameters.

\paragraph{Acknowledgments}
\noindent We thank the anonymous referees for helpful comments.  René van Bevern
was supported by the DFG, project DAPA (NI 369/12), Robert Bredereck
by the DFG, project PAWS (NI 369/10), Jiehua Chen by the
Studienstiftung des Deutschen Volkes, Vincent Froese by the DFG,
project DAMM (NI 369/13), and Gerhard J.\ Woeginger, while visiting
TU~Berlin, by a Humboldt Research Award of the Alexander 
von Humboldt Foundation, Bonn, Germany.

\bibliographystyle{abbrvnat}
\bibliography{gerry}

\end{document}